\def\draft{0}
\documentclass[11pt,letterpaper]{article}
\usepackage[utf8]{inputenc}
\usepackage{fullpage}
\usepackage{amsthm}
\usepackage{xcolor}
\usepackage{mathtools,amsmath,amssymb}
\usepackage{mathrsfs}
\usepackage{algorithm, algpseudocode}
\usepackage[colorlinks=true, allcolors=blue]{hyperref}
\usepackage[capitalise,nameinlink]{cleveref}
\usepackage[style=numeric-comp, backend=biber, minalphanames=3, maxalphanames=4, maxbibnames=99, sorting=none]{biblatex}

\newcommand{\gnote}[1]{\ifnum\draft=1 {\color{red} [\textbf{P:} #1]}\fi}
\newcommand{\vnote}[1]{\ifnum\draft=1 {\color{magenta} [\textbf{S:} #1]}\fi}
\newcommand{\knote}[1]{\ifnum\draft=1 {\color{orange} [\textbf{K:} #1]}\fi}

\usepackage{bm}
\usepackage{enumitem}
\setlist[itemize]{topsep=3pt}
\setlist[enumerate]{topsep=3pt}
\usepackage{mathtools}
\usepackage{complexity}
\usepackage{dsfont}
\usepackage{float}
\usepackage{thm-restate}
\usepackage{bbm}
\usepackage{amsthm}
\usepackage{thmtools,thm-restate}

\numberwithin{equation}{section}

\usepackage[colorlinks=true, allcolors=blue]{hyperref}
\usepackage[nameinlink,capitalise]{cleveref}
\hypersetup{
    citecolor={violet}
}

\declaretheoremstyle[bodyfont=\it,qed=\qedsymbol]{noproofstyle}

\declaretheorem[name=Observation,numbered=no]{observation*}

\declaretheorem[numberlike=equation]{theorem}

\declaretheorem[name=Theorem,numbered=no]{theorem*}

\declaretheorem[numberlike=equation]{lemma}
\declaretheorem[name=Lemma,numbered=no]{lemma*}

\declaretheorem[numberlike=equation]{corollary}
\declaretheorem[name=Corollary,numbered=no]{corollary*}

\declaretheorem[name=Proposition,numbered=no]{proposition*}

\declaretheorem[name=Claim,numbered=no]{claim*}

\declaretheorem[name=Conjecture,numbered=no]{conjecture*}

\declaretheorem[name=Question,numbered=no]{question*}

\declaretheoremstyle[bodyfont=\it]{defstyle} 

\declaretheorem[numberlike=equation,style=defstyle]{definition}
\declaretheorem[unnumbered,name=Definition,style=defstyle]{definition*}

\declaretheorem[unnumbered,name=Notation=defstyle]{notation*}

\declaretheorem[unnumbered,name=Construction,style=defstyle]{construction*}

\declaretheoremstyle[]{rmkstyle} 

\declaretheorem[numberlike=equation,style=rmkstyle]{example}
\declaretheorem[unnumbered,name=Example,style=rmkstyle]{example*}

\declaretheorem[unnumbered,name=Remark,style=rmkstyle]{remark*}

\renewcommand{\b}[1]{\boldsymbol{#1}}

\newcommand{\sfv}[0]{\mathsf{v}}
\newcommand{\sfV}[0]{\mathsf{V}}
\newcommand{\sfi}[0]{\mathtt{i}}
\newcommand{\sfj}[0]{\mathtt{j}}
\newcommand{\sfk}[0]{\mathtt{k}}
\newcommand{\sfx}[0]{\mathsf{x}}

\usepackage[margin=1in]{geometry}
\usepackage[compact]{titlesec}
\titlespacing{\section}{0pt}{1.5ex}{0ex}
\titlespacing{\subsection}{0pt}{1.5ex}{0ex}
\titlespacing{\subsubsection}{0pt}{1ex}{0ex}
\titlespacing{\paragraph}{0pt}{1.5ex}{1ex}
\title{Fair allocation of a multiset of indivisible items}

\author{Pranay Gorantla\thanks{Physics Department, Princeton University, Princeton, NJ, USA.
Email: \texttt{gorantla@princeton.edu}}
\and Kunal Marwaha\thanks{Department of Computer Science, University of Chicago, Chicago, IL, USA.
Supported by the National Science Foundation Graduate Research Fellowship Program under Grant No.
DGE-1746045.
Email: \texttt{kmarw@uchicago.edu}}
\and Santhoshini Velusamy\thanks{School of Engineering and Applied Sciences, Harvard University, Cambridge, MA, USA.
Supported in part by a Google Ph.D. Fellowship, a Simons 
Investigator Award to Madhu Sudan, and NSF Awards CCF 1715187 and CCF 2152413.
Email: \texttt{svelusamy@g.harvard.edu}}
}

\begin{document}
\date{}
\maketitle

\begin{abstract}

We study the problem of \emph{fairly} allocating a \emph{multiset} $M$ of $m$ \emph{indivisible} items among $n$ agents with \emph{additive} valuations. 
Specifically, we introduce a parameter $t$ for the number of distinct \emph{types} of items and study fair allocations of multisets that contain only items of these $t$ types, under two standard notions of fairness:
\begin{enumerate}
    \item Envy-freeness (EF): For arbitrary $n$, $t$, we show that a complete EF allocation exists when at least one agent has a \emph{unique} valuation and the number of items of each type exceeds a particular finite \emph{threshold}. We give explicit upper and lower bounds on this threshold in some special cases.
    \item Envy-freeness up to any good (EFX): For arbitrary $n$, $m$, and for $t\le 2$, we show that a complete EFX allocation always exists. We give two different proofs of this result.
    One proof is constructive and runs in polynomial time; the other is geometrically inspired.
\end{enumerate}

\end{abstract}

\clearpage
\pagebreak
\tableofcontents
\pagebreak

\section{Introduction}\label{sec:intro}
Fair allocation of \emph{indivisible} items is a well-studied, fundamental problem in economics.
Given a set $M$ of $m$ items, and $n$ agents with individual valuations, the goal is to \emph{completely} allocate all the items among the agents in a \emph{fair} manner.
Several notions of fairness have been considered in the literature.
One of the most well-studied notions is \emph{envy-freeness (EF)}.
An allocation is considered to be an EF allocation if no agent \emph{envies} another agent, i.e., prefers the set of items received by another agent to their own set.
Although such an allocation might be ideally ``fair'', complete EF allocations do not always exist.
For example, say we allocate one item among two agents; certainly the agent who does not get the item is envious.

There are several ways to relax the definition of envy-freeness.
One such definition is 
\emph{envy-freeness up to one item (EF1)}, introduced by Budish in \cite{ef1_defn}.
In an EF1 allocation, agents are allowed to be envious, but the envy of any agent disappears when \emph{some} item is removed from the set of the envied agent.
A complete EF1 allocation always exists and can be obtained in polynomial time using the ``envy-cycles'' procedure by Lipton, Markakis, Mossel, and Saberi \cite{envy_cycles}.
In fact, in the special case when all agents have \emph{additive} valuations, there is a simpler procedure: order the agents arbitrarily, and keep allocating their most preferred among the remaining items one-by-one in a round-robin way.

Unfortunately, there are many settings where an EF1 allocation is not intuitively fair.
For example, say we have two agents whose valuations are additive, and three items $g_1$, $g_2$, and $g_3$ that are valued as 10, 4, and 5 respectively by both the agents.
While the allocation where one agent gets $\{g_1,g_2\}$ and the other agent gets $\{g_3\}$ is an EF1 allocation, an intuitively fairer allocation is one where one agent gets $\{g_1\}$ and the other agent gets $\{g_2,g_3\}$.
This disparity inspired the definition of a stronger notion of fairness, \emph{envy-freeness up to any item (EFX)}, by Caragiannis, Kurokawa, Moulin, Procaccia, Shah, and Wang in \cite{efx_nashwelfare}.
An allocation is said to be an EFX allocation if no agent envies any strict subset of items received by any other agent.
Note that by definition, an EFX allocation is also an EF1 allocation.
In the example stated above, the only complete EFX allocations are $\langle\{g_1\},\{g_2,g_3\} \rangle$ and $\langle\{g_2,g_3\},\{g_1\} \rangle$.
Despite significant efforts, complete EFX allocations have been shown to exist only when valuations are identical \cite{efx_leximin_general,efx_charity}, agents have identical additive \emph{preferences} \cite{efx_leximin_general}, $n=2$ \cite{efx_leximin_general}, or $n=3$ with additive valuations \cite{efx_3agents}.
The latter result was recently extended to settings where two out of three valuations are general \emph{monotone} functions \cite{EFX_latest}.
It is still open whether EFX exists for all sets of items and valuations.

In this paper, we reformulate the fair allocation problem to allow $M$ to be a multiset.
Specifically, we introduce a parameter $t$ for the number of distinct \emph{types} of items. There are several natural settings where $t$ is small while $m$ and $n$ are arbitrarily large.
For example, consider a scenario where a local farm has excess produce and wants to distribute them fairly among their employees.
The number of distinct types of produce is usually small, while the excess produce and the number of employees can be quite large.
This is in contrast to the setup in many previous works which restrict $n$ to be small \cite{efx_leximin_general,efx_3agents,efx_4agents}.

In what follows, we precisely define the fair allocation problem and the notions of fairness that we consider, and then state our results.

\subsection{Setup}\label{sec:setup}

Let $t$ denote the number of distinct types of items, and let $g_a$ denote the item of type $a\in[t]:=\{1,\ldots,t\}$.
Let $\mathscr M$ be the collection of all multisets of items of these $t$ types.
Consider any multiset $X:=\{g_a^{x_a}:a\in[t]\}\in\mathscr M$, where $x_a\ge0$ denotes the multiplicity of $g_a$ (i.e., number of items of type $a$) in $X$.
The multiset $X$ is associated with a nonnegative integer point $\b x := (x_1,\ldots,x_t) = \sum_{a\in[t]}x_a \b e_a\in \mathbb Z_{\ge0}^t$, where $\b e_1,\ldots,\b e_t$ are the standard unit vectors of $\mathbb R^t$.
This gives a one-to-one correspondence between $\mathscr M$ and $\mathbb Z_{\ge 0}^t$.

Standard operations on multisets $X=\{g_a^{x_a}:a\in[t]\}$ and $Y=\{g_a^{y_a}:a\in[t]\}$ include:
\begin{itemize}
	\item \emph{Containment}: An item $g$ of type $a\in[t]$ is contained in $X$, i.e., $g\in X$, if $x_a>0$.
	\item \emph{Inclusion}: $X$ is a subset of $Y$, i.e., $X\subseteq Y$, if $x_a\le y_a$ for all $a\in[t]$. We also denote this as $\b x\le \b y$. The inclusion is strict, i.e., $X\subsetneq Y$, if in addition $x_a < y_a$ for some $a\in[t]$. We also denote this as $\b x<\b y$.
	\item \emph{Union}: The union of $X$ and $Y$ is $X\cup Y:=\{ g_a^{\max\{x_a,y_a\}}:a\in[t]\}$.
	\item \emph{Intersection}: The intersection of $X$ and $Y$ is $X\cap Y:=\{ g_a^{\min\{x_a,y_a\}}:a\in[t]\}$.
	\item \emph{Sum}: The sum of $X$ and $Y$ is $X\uplus Y:=\{ g_a^{x_a+y_a}:a\in[t]\}$.
	\item \emph{Difference}: The difference of $X$ and $Y$ is $X\setminus Y:=\{ g_a^{\max\{x_a-y_a,0\}}:a\in[t]\}$.
\end{itemize}
In this paper, we need only the first two and the last two operations.
Intuitively, for any multiset $X$ and any item $g$, the sum $X\uplus\{g\}$ is the same as \emph{adding} the item $g$ to $X$ and the difference $X\setminus\{g\}$ is the same as \emph{removing} the item $g$ from $X$ if $g\in X$.

There are $n$ agents.
Each agent $i\in [n]$ has a valuation $V_i:\mathbb Z_{\ge0}^t \rightarrow [0,\infty)$ that satisfies \emph{monotonicity}, i.e., $V_i(\b x) \le V_i(\b y)$ if $\b x \le \b y$.
We use $V_i(\b x)$ and  $V_i(X)$ interchangeably because of the correspondence between $\mathbb Z_{\ge 0}^t$ and $\mathscr M$.
Then, monotonicity can be stated as $V_i(X) \le V_i(Y)$ if $X\subseteq Y$.
The valuation $V_i$ is said to be \emph{additive} if for every $\b x\in\mathbb Z_{\ge 0}^t$, $V_i(\b x):=\sum_{a\in[t]} v_i(a) x_a$, where the function $v_i:[t]\rightarrow [0,\infty)$ is such that $v_i(a)$ is agent $i$'s value for item $g_a$.\footnote{We assume that every agent has a positive value for at least one type of item, i.e., $v_i$ is not the zero function; otherwise, even if no items are allocated to them, they do not envy anyone.}
In particular, an additive valuation is described by only $t$ parameters.
We refer to $v_i$ as the \emph{item-value function} associated with the valuation $V_i$.

Consider a multiset $M:=\{g_a^{m_a}:a\in[t]\}\in\mathscr M$ associated with $\b m := (m_1,\ldots,m_t)\in\mathbb Z_{\ge0}^t$.
We want to allocate these $m:=|M|=\sum_{a\in[t]}m_a$ items to the $n$ agents.
An allocation $\mathcal X:=\langle X_1,X_2,\ldots,X_n\rangle$ is an ordered partition of a multiset $M'\subseteq M$ into $n$ multisets, i.e., $\biguplus_{i\in[n]}X_i = M'$, such that $X_i$ is allocated to agent $i$.
An allocation is said to be \emph{partial} if $M'\subsetneq M$, and \emph{complete} if $M'=M$.
We say that an allocation $\mathcal X$ is
\begin{itemize}
    \item an EF allocation if for every $i\neq j$, $V_i(X_i)\ge V_i(X_j)$,
    \item an EF1 allocation if for every $i\neq j$, there is a $g\in X_j$ such that $V_i(X_i)\ge V_i(X_j\setminus\{g\})$, and
    \item an EFX allocation if for every $i\neq j$, and any $g\in X_j$, $V_i(X_i)\ge V_i(X_j\setminus\{g\})$.
\end{itemize}
Observe that $\mathcal X$ is an EFX allocation if and only if for every $i\neq j$, and any $S\subsetneq X_j$, $V_i(X_i)\ge V_i(S)$.
It is also worth noting that when agents have additive valuations, $\mathcal X$ is an EFX allocation if and only if for every $i\ne j$, $V_i(X_i)\ge V_i(X_j\setminus\{g\})$, where $g$ is the least preferred item of $i$ in $X_j$.

We note that when $t=m$, $M$ is a set, and we recover the original fair allocation problem.
\subsection{Our results}

\subsubsection{Result 1: EF with enough items}\label{sec:result1}

We first prove that if at least one agent has a \emph{unique} additive valuation, then a complete EF allocation always exists as long as there are \emph{enough} items of each type.
While complete EF allocations are known to exist with high probability in certain \emph{randomized} settings \cite{EFX_random_14,EFX_random_20,EFX_random_21,EF_smoothed_utilities}, ours is the first result that shows existence of complete EF allocations in a deterministic setting.

Let us define what we mean by an agent having a unique additive valuation.
\begin{definition}\label{def:identical-distinct}
The additive valuations $V_i$ and $V_j$ are said to be \emph{identical} if there exists a $\lambda >0$ such that $v_i(a) = \lambda v_j(a)$, for all $a\in [t]$.
They are said to be \emph{distinct} otherwise.
\end{definition}

An agent is said to have a unique valuation if their valuation is distinct from every other agent's valuation.
We provide a formal statement of our result below:
\begin{theorem}\label{thm:intro_EF}
If the valuations $V_1,\dots,V_n$ of the agents are additive and at least one of them is unique, then there exists a $\mu\in \mathbb{N}$ such that whenever there are at least $\mu$ items of each type in $M$, i.e., $\b m \ge \mu \b 1^t$, there exists a complete EF allocation of $M$.\footnote{We use the notation $\b 0^q := (0,\ldots,0)\in \mathbb R^q$, and $\b 1^q := (1,\ldots,1)\in \mathbb R^q$ for any positive integer $q$.}
\end{theorem}

In fact, we prove the following stronger theorem in \cref{sec:proof_EF} that implies \cref{thm:intro_EF}.
Let $1\le d \le n$ and $\sfV_1,\dots,\sfV_d$ be pairwise distinct additive valuations.
Let $n_1,\dots,n_d$ be positive integers with $\sum_{\sfi=1}^d n_\sfi = n$.
Say there are $n_\sfi$ agents with valuations identical to $\sfV_\sfi$ for all $\sfi\in [d]$. Let $r := \gcd(n_1,\dots,n_d)$.

\begin{theorem}\label{thm:intro_EF_0modr}
There exists a $\mu\in \mathbb{N}$ such that whenever $\b m \ge \mu \b 1^t$ and $\b m \equiv \b 0^t \pmod r$, there exists a complete EF allocation of $M$.
\end{theorem}

The two corollaries below directly follow from the above theorem, and the first corollary implies \cref{thm:intro_EF}.
\begin{corollary}[EF with enough items]\label{cor:EF_enoughitems}
If $r = 1$, then there exists a $\mu\in \mathbb{N}$ such that whenever $\b m \ge \mu \b 1^t$, there exists a complete EF allocation of $M$.
\end{corollary}

\begin{corollary}[EF with \emph{charity}]\label{cor:EFcharity}
There exists a $\mu\in \mathbb{N}$ such that whenever $\b m \ge \mu \b 1^t$, there exists an EF allocation of $M$ with at most $(r-1)$ unallocated items of each type.
\end{corollary}

\cref{cor:EFcharity} is similar in spirit to the results in \cite{efx_charity,efx_additive_to_general}: for agents with any general monotone valuations, there always exists an EFX allocation of $M$ with at most $n-1$ (improved to $n-2$ in \cite{efx_additive_to_general}) unallocated items.

In \cref{sec:bound_mu}, we prove upper and lower bounds on the value of $\mu$ in \cref{thm:intro_EF_0modr} in some special cases:
\begin{itemize}
    \item For $d= 2$, we show that $\mu = O(n^2\sqrt{t}/(r\delta))$, where $\delta$ is a measure of the ``distinctness'' of the valuations $\sfV_1,\dots,\sfV_d$\footnote{See \cref{sec:upperbound_mu} for the formal definition of $\delta$.}; we also show that $\mu= \Omega(n^2/(r\delta))$.
    \item For $d>2$ and $t=2$, we show that $\mu = O(n^2/(r\delta))$ for $t=2$.
\end{itemize}
Observe that for $d=1$, we have $r=n$. In this case, \cref{thm:intro_EF_0modr} is trivial as a complete EF allocation always exists when $\b m \equiv \b 0^t \pmod n$.

We prove a partial converse of \cref{thm:intro_EF_0modr} in \cref{sec:result1_converse}. In particular, we show that there \emph{exist} valuations for which whenever $\b m \not\equiv \b 0^t \pmod r$, there is no complete EF allocation.
We also prove a partial converse of \cref{cor:EF_enoughitems}. Namely, we show that there exist valuations for which if $r\ne 1$, then complete EF allocations need not exist even where there is an arbitrarily large number of items of each type. We also prove a full converse of \cref{cor:EF_enoughitems} in the special case of $d=1$ in \cref{sec:fullConverse}. Namely, when $n\ge 2$ and all agents have identical valuations, we show that complete EF allocations need not exist even when there is an arbitrarily large number of items of each type.

We now give an overview of the techniques involved in the proof of \cref{thm:intro_EF_0modr}. We first formulate the EF allocation problem as an \emph{integer linear programming} problem in $nt$ variables.
Let $V_1,\dots,V_n$ be the valuations of the agents and let $v_1,\dots,v_n$ be the corresponding item-value functions.
Let $x_{i,a}$ denote the number of items of type $a$ allocated to agent $i$.
It is a complete EF allocation if and only if
\begin{align*}
    &\sum_{a\in[t]} v_i(a) x_{i,a} \ge \sum_{a\in[t]} v_i(a) x_{j,a}~,\quad \forall i,j\in[n]~,\quad i\ne j~,
    \\
    &x_{i,a}\ge 0~,\quad\forall i\in[n]~,~a\in[t]~,
    \\
    &\sum_{i\in[n]} x_{i,a} = m_a~,\quad \forall a\in[t]~,
    \\
    &x_{i,a} \in\mathbb Z~,\quad \forall i\in[n]~,~a\in[t]~.
\end{align*}
Let $\b v:=(v_1,\ldots,v_n)$, and let $\Delta(\b v,\b m)\subseteq \mathbb R^{nt}$ denote the feasible set of the LP relaxation of this problem with the integrality constraints removed.
Our goal then is to show that $\Delta(\b v,\b m)$ contains an integer point under the assumptions of \cref{thm:intro_EF_0modr}.

However, even if we could show that $\Delta(\b v,\b m)$ contains an integer point for some $\b m \equiv \b 0^t \pmod r$, it is not immediately clear how to show that $\Delta(\b v,\b m')$ contains an integer point for all $\b m' \ge \b m$ with $\b m' \equiv \b 0^t \pmod r$.
This is because $\Delta(\b v, \b m)$ is not contained in $\Delta(\b v, \b m')$ in general.

To overcome this problem, we reformulate the LP problem in terms of a new set of variables, $\xi_{\sfi,a}$ for $\sfi\in[d-1]$ and $a\in[t]$.
Let $\b\sfv:= (\sfv_1,\dots,\sfv_d)$ be the item-value functions associated with $\sfV_1\dots,\sfV_d$. 
We denote the feasible set of this new LP problem by $\Gamma(\b \sfv,\b m)\subseteq \mathbb R^{(d-1)t}$.
We show that it has the following nice properties:
\begin{itemize}
    \item If $\b m\equiv \b0^t\pmod r$ and $\Gamma(\b \sfv,\b m)$ contains an ``integer cube'' (a cube with integer points as corners) of length $\frac{n}{r}-1$,\footnote{Since $r=\gcd(n_1,\dots,n_d)$, it follows that $r$ divides $n=n_1+\cdots+n_d$. Hence, $\frac{n}{r}$ is an integer.} with edges parallel to the coordinate axes, then one of the integer points in the cube maps to an integer point in $\Delta(\b v,\b m)$.
    \item $\Gamma(\b \sfv,\b m)$ satisfies the ``containment'' property: if $\b m' \ge \b m$, then $\Gamma(\b \sfv, \b m')$ contains $\Gamma(\b \sfv,\b m)$.
\end{itemize}
Combining these properties, it suffices to show that $\Gamma(\b \sfv,\mu \b 1^t)$ contains such an integer cube for sufficiently large $\mu$.
The bulk of \cref{sec:proof_EF} is devoted to showing precisely this.

\subsubsection{Result 2: EFX with at most two types of items}\label{sec:result3}

Our second result is the following:
\begin{theorem}\label{thm:intro_t2}
Complete EFX allocations exist when agents have additive valuations and the number of distinct types of items is at most two, i.e., $t\le 2$.
\end{theorem}

When $t=1$, there is a simple way to allocate the $m=m_1$ items.
First allocate $\lfloor m/n \rfloor$ items to every agent.
Then allocate the remaining $u := m - n\lfloor m/n \rfloor <n$ unallocated items equally among an arbitrary choice of $u$ agents.
Although these $u$ agents are envied, removing any item from them removes the envy.
Hence, this is a complete EFX allocation.

For $t=2$, our algorithm is more involved (see \cref{alg:tequals2}).
In each step, we allocate items one-by-one such that the allocation at the end of that step is EFX.
We divide the agents into two sets: $N_1$ is the set of agents who prefer a type-$1$ item to a type-$2$ item, and $N_2:=[n]\setminus N_1$ is the set of remaining agents. We briefly outline the four steps of the algorithm:
\begin{enumerate}
	\item In the first step, we give all agents their most preferred items in a round-robin way.
	Each round in this step consists of allocating one item to every agent.
	The step ends when there are not enough unallocated items to complete a round.
	Observe that the allocation at the end of this step is EF because all agents have equal number of items, and every item they have is their most preferred.
	
	\item Say fewer than $|N_1|$ type-$1$ unallocated items are leftover after the previous step.\footnote{If not, then there must be fewer than $|N_2|$ type-$2$ unallocated items. In this case, the rest of the algorithm is the same with the types $1$ and $2$ exchanged.}
	In the second step, we give these items, one each, to the agents who most value a type-$1$ item over a type-$2$ item.
	We denote the set of these agents as $N_1^+$, and define $N_1^-:=N_1\setminus N_1^+$.
	Observe that removing any item from an agent in $N_1^+$ removes any envy towards them, and hence the resulting allocation is EFX.
	
	\item There are only type-$2$ unallocated items leftover after the previous step. In the third step, we give them one-by-one to all the agents in $[n]\setminus N_1^+$ in a round-robin way.
	At the end of every round, we check if some agent in $N_1^+$ envies some agent. If so, the step ends, else, we proceed to the next round.
	Roughly, the idea is that agents in $[n]\setminus N_1^+$ are better off at the end of this step compared to the previous step, and hence their envy towards agents in $N_1^+$ only decreases.
	(In fact, if this step ends with type-$2$ items remaining and an agent in $N_1^+$ envying some agent, we show that no agent in $[n]\setminus N_1^+$ envies any agent.
	The proof of this crucially uses the fact that agents in $N_1^+$ are those who most value a type-$1$ item over a type-$2$ item.)
	On the other hand, since this step ends as soon as an agent in $N_1^+$ envies some other agent, and their least preferred (type-$2$) item was the last item to be allocated, removal of this item removes the envy, thus resulting in an EFX allocation.
	The step can also end if we run out of unallocated items.
	Even in this case, we show that the resulting allocation is EFX.
	
	\item In the fourth step, we simply allocate the leftover type-$2$ items in a round-robin way to all agents, starting with those in $N_1^+$.
	To show that the resulting allocation is EFX, we crucially rely on the fact (mentioned above) that agents in $[n]\setminus N_1^+$ do not feel envy initially, and agents in $N_1^+$ lose any envy as soon as they receive a type-$2$ item.
\end{enumerate}

We give formal pseudocode for the algorithm and a proof of correctness in \cref{sec:t=2}.

In \cref{sec:geometric}, we give an alternative geometrical proof of \cref{thm:intro_t2} and discuss roadblocks to generalizing this proof for $t\ge 3$. Our geometrical proof for $t=2$ uses well-known ideas introduced in \cite{efx_charity}---such as \emph{envy graph}, \emph{Pareto dominance}, \emph{most envious agent}, and \emph{reachability}.
In particular, we show that there is a source in the envy graph with a most envious agent that is reachable from the source.
To show this, we introduce new geometrical notions where a hyperplane in $\mathbb R^t$ represents the valuation of an agent, and a point on this hyperplane represents the set of items allocated to this agent.

\subsection{Future questions}

There are several directions in which our results can be extended. We list a few of them here.
\begin{enumerate}
    \item Can we obtain explicit upper and lower bounds on the threshold $\mu$ in \cref{thm:intro_EF} when both $d$ and $t$ are arbitrary?
    
    \item Can we obtain a complete characterization of valuations $V_1,\dots,V_n$ for which a complete EF allocation exists when there are enough items of each type?
    
    \item Do complete EFX allocations exist for $t\ge3$ when agents have additive valuations?
\end{enumerate}

\subsection{Related work}
Fair division has received a lot of attention since the introduction of the cake cutting problem by Steinhaus in \cite{cake_cutting}.
While there are finite bounded cake cutting protocols that guarantee \emph{envy-freeness} for any number of agents \cite{ef_protocol95,ef_protocol97,ef_protocol00,ef_protocol16,ef_anyagents}, envy-free (EF) allocations do not always exist when items are indivisible.
In fact, the problem of deciding whether or not a complete EF allocation exists is known to be NP-complete \cite{envy_cycles}.
However, complete EF allocations are known to exist with high probability in randomized settings, where the agents have additive valuations drawn from independent probability distributions and the number of items is sufficiently large relative to the number of agents \cite{EFX_random_14,EFX_random_20,EFX_random_21}.
A more recent work \cite{EF_smoothed_utilities} considers a smoothed model where each agent's item-values are independently and randomly perturbed.
They show that with sufficiently many items, a complete EF allocation exists with high probability if a large enough fraction of item-values are perturbed.

Many recent works study the allocation of items under relaxed notions of envy-freeness (including EF1 and EFX) in a number of settings \cite{EFX_cite_0,efx_nashwelfare,EFX_cite_2,EFX_cite_3,EFX_cite_4,EFX_cite_00,EFX_cite_5,efx_leximin_general,efx_charity,efx_3agents,EFX_latest,efx_4agents,efx_two_valuations, efx_additive_to_general,EFX_cite_6,EFX_cite_7,EFX_cite_8,EFX_cite_9,EFX_cite_10,EFX_cite_11,efx_with_copies,epistemicEFX}.
The setting considered in \cite{efx_with_copies} is superficially related to our setting in that they too consider ``copies'' of items.
However, they restrict allocations to those where no agent gets more than one copy of each item.
This restriction is so strong that even for three agents with identical additive valuations, they show that a complete EFX allocation need not exist among this restricted set of allocations.\footnote{Recall that complete EFX allocations do exist for three agents with additive valuations \cite{efx_3agents} and $n$ agents with identical valuations \cite{efx_leximin_general}.}

\section{Existence of complete EF allocations when items are plentiful} \label{sec:EF}

In this section, we prove \cref{thm:intro_EF_0modr}.
Let $V_1,\dots,V_n$ be the valuations of the agents and let $v_1,\dots,v_n$ be the corresponding item-value functions.
Let $x_{i,a}$ denote the number of items of type $a$ allocated to agent $i$.
It is a complete EF allocation if and only if
\begin{align}
    &\sum_{a\in[t]} v_i(a) x_{i,a} \ge \sum_{a\in[t]} v_i(a) x_{j,a}~,\quad \forall i,j\in[n]~,\quad i\ne j~,\label{EF-def}
    \\
    &x_{i,a}\ge 0~,\quad\forall i\in[n]~,~a\in[t]~,\label{pos-def}
    \\
    &\sum_{i\in[n]} x_{i,a} = m_a~,\quad \forall a\in[t]~,\label{comp-def}
    \\
    &x_{i,a} \in\mathbb Z~,\quad \forall i\in[n]~,~a\in[t]~.\label{int-def}
\end{align}
We call these the \emph{EF} constraints, the \emph{positivity} constraints, the \emph{completeness} constraints, and the \emph{integrality} constraints respectively.
Let $\b v:=(v_1,\ldots,v_n)$, and let $\Delta(\b v,\b m)\subseteq \mathbb R^{nt}$ denote the feasible set of the LP relaxation of this problem with the integrality constraints removed.

Let $\sfV_1,\dots,\sfV_d$ be pairwise distinct additive valuations and let $\sfv_1,\dots,\sfv_d$ be the corresponding item-value functions. For each $\sfi\in[d]$, let $N_\sfi$ be the set of agents with valuations identical to $\sfV_\sfi$.
In total, there are $n=\sum_{\sfi=1}^d n_\sfi$ agents, where $n_\sfi:=|N_\sfi|\ge1$. Define $r:=\gcd(n_1,\ldots,n_d)$. The main result of this section is the following:
\begin{theorem}\label{thm:EF}
There is a $\mu\in\mathbb N$ such that as long as $\b m\ge \mu\b 1^t$ and $\b m \equiv \b 0^t \pmod r$, the feasible set $\Delta(\b v,\b m)$ contains an integer point.\footnote{Recall that we use the notation $\b 0^q:=(0,\ldots,0)\in \mathbb R^q$, and $\b 1^q:=(1,\ldots,1)\in \mathbb R^q$ for any positive integer $q$.}
\end{theorem}
This is a restatement of \cref{thm:intro_EF_0modr}.
We prove \cref{thm:EF} in the next subsection.

\subsection{Proof of \cref{thm:EF}}\label{sec:proof_EF}
Label the agents so that $N_1 = \{1,\ldots,n_1\}$, $N_2 = \{n_1+1,\ldots,n_1+n_2\}$, and so on. For any $\sfi\in[d]$, if $i,j\in N_\sfi$, then the EF inequalities \eqref{EF-def} between agents $i$ and $j$ become equations because they have identical valuations. One way to satisfy these equations is by choosing $x_{i,a}=x_{j,a}$ for all $a\in[t]$. 
In fact, it suffices to prove \cref{thm:EF} by restricting to such solutions.

For any $\sfi\in[d]$ and $a\in[t]$, let $\sfx_{\sfi,a}$ denote the number of type-$a$ items allocated to each agent in $N_\sfi$. In other words, for all $i\in N_\sfi$, we have $x_{i,a} = \sfx_{\sfi,a}$. The remaining constraints of \eqref{EF-def}, \eqref{pos-def} and \eqref{comp-def} can be written as
\begin{align}
    &\sum_{a\in[t]} \sfv_\sfi(a) \sfx_{\sfi,a} \ge \sum_{a\in[t]} \sfv_\sfi(a) \sfx_{\sfj,a}~,\quad \forall \sfi,\sfj\in[d]~,\quad \sfi\ne \sfj~,\label{EF-def'}
    \\
    &\sfx_{\sfi,a}\ge 0~,\quad\forall \sfi\in[d]~,~a\in[t]~,\label{pos-def'}
    \\
    &\sum_{\sfi\in[d]} n_\sfi \sfx_{\sfi,a} = m_a~,\quad \forall a\in[t]~,\label{comp-def'}
\end{align}
Let us define a new set of variables, $\xi_{\sfi,a} := \sfx_{\sfi,a}-\sfx_{\sfi+1,a}$ for $\sfi\in[d-1]$ and $a\in[t]$, so that we can write $\sfx_{\sfi,a}=\sfx_{d,a} + \sum_{\sfk=\sfi}^{d-1} \xi_{\sfk,a}$ for any $\sfi\in[d-1]$ and $a\in[t]$.
Note that this is a \emph{unimodular} transformation from $\sfx_{\sfi,a}$ to $\xi_{\sfi,a}$ and $\sfx_{d,a}$.\footnote{A unimodular transformation is an integer linear transformation with determinant $1$.
In particular, it maps $\mathbb{Z}^{dt}$ to itself bijectively.}
In terms of $\xi_{\sfi,a}$ and $\sfx_{d,a}$, the constraints of \eqref{EF-def'}, \eqref{pos-def'} and \eqref{comp-def'} can be written as
\begin{align}
    &\sum_{a\in[t]} \sfv_\sfi(a) \sum_{\sfk=\sfi}^{\sfj-1} \xi_{\sfk,a} \ge 0~,\quad\text{and}\quad \sum_{a\in[t]} \sfv_\sfj(a) \sum_{\sfk=\sfi}^{\sfj-1} \xi_{\sfk,a} \le 0~, \quad \forall 1\le \sfi<\sfj\le d~,\label{EF'}
    \\
    &\sfx_{d,a} + \sum_{\sfk=\sfi}^{d-1} \xi_{\sfk,a} \ge 0~,\quad \forall \sfi\in[d]~,~a\in[t]~,\label{pos'}
    \\
    &\sum_{\sfk=1}^{d-1} \left(\sum_{\sfi=1}^{\sfk} n_\sfi\right) \xi_{\sfk,a} + n \sfx_{d,a} = m_a~,\quad \forall a\in[t]~.\label{comp'}
\end{align}
We solve for $\sfx_{d,a}$ using the completeness constraints \eqref{comp'}:
\begin{equation*}
    \sfx_{d,a}(\b \xi) = \frac1n \left[m_a - \sum_{\sfk=1}^{d-1} \left(\sum_{\sfi=1}^{\sfk} n_\sfi\right) \xi_{\sfk,a} \right]~,
\end{equation*} where $\b \xi := (\xi_{\sfi,a})$ denotes the remaining $(d-1)t$ variables.
We are then left with only the EF and positivity constraints in $(d-1)t$ variables,
\begin{align}
    &\sum_{a\in[t]} \sfv_\sfi(a) \sum_{\sfk=\sfi}^{\sfj-1} \xi_{\sfk,a} \ge 0~,\quad\text{and}\quad \sum_{a\in[t]} \sfv_\sfj(a) \sum_{\sfk=\sfi}^{\sfj-1} \xi_{\sfk,a} \le 0~, \quad \forall 1\le \sfi<\sfj\le d~,\label{EF}
    \\
    &m_a - \sum_{\sfk=1}^{\sfi-1} \left(\sum_{\sfj=1}^{\sfk} n_\sfj\right) \xi_{\sfk,a} + \sum_{\sfk=\sfi}^{d-1} \left(n-\sum_{\sfj=1}^{\sfk} n_\sfj\right) \xi_{\sfk,a} \ge 0~,\quad \forall \sfi\in[d]~,~a\in[t]~.\label{pos}
\end{align}
Define $\b \sfv := (\sfv_1,\ldots,\sfv_d)$, and let $\Gamma(\b \sfv,\b m)\subseteq \mathbb R^{(d-1)t}$ denote the feasible set of \eqref{EF} and \eqref{pos}.
Note that not every integer point in $\Gamma(\b \sfv,\b m)$ maps to an integer point in $\Delta(\b v,\b m)$ because $\sfx_{d,a}(\b \xi)$ need not be an integer for every integer point $\b \xi\in \Gamma(\b \sfv, \b m)$.
Nonetheless, the following is true:

\begin{lemma}\label{lem:int-cube}
If $\b m\equiv \b0^t\pmod r$, and if there is an integer point $\b \xi^*\in \Gamma(\b \sfv,\b m)$ such that $\b \xi^* + \b c\in \Gamma(\b \sfv,\b m)$ for all $\b c\in [-\frac{14 n}{rd},\frac{14 n}{rd}]^{(d-1)t} \cap \mathbb Z^{(d-1)t}$,\footnote{Since $r=\gcd(n_1,\dots,n_d)$, we have $n_\sfi \ge r$ for all $\sfi\in[d]$, and so $n\ge rd$. Hence, $\frac{14 n}{rd} \ge 14>1$.} then there is a $\b c^*\in [-\frac{14 n}{rd},\frac{14 n}{rd}]^{(d-1)t} \cap \mathbb Z^{(d-1)t}$ such that $\sfx_{d,a}(\b \xi^* + \b c^*)$ is an integer for all $a\in[t]$.
\end{lemma}
The hypothesis of \cref{lem:int-cube} is equivalent to the statement that $\Gamma(\b \sfv,\b m)$ contains a $(d-1)t$-dimensional ``integer cube'' (cube with integer points as corners) of length $2\lfloor\frac{14 n}{rd}\rfloor$, with edges parallel to the coordinate axes, centered at an integer point $\b \xi^*\in\Gamma(\b \sfv,\b m)$. So \cref{lem:int-cube} ensures that one of the integer points in this cube maps to an integer point in $\Delta(\b v,\b m)$. In order to prove \cref{lem:int-cube}, we need the following mathematical fact proved in \cite{gcd_result}:
\begin{lemma}\label{lem:gcd_result}
For any integer $q$, there exist integers $Y_1,\ldots,Y_d$ with $|Y_\sfi| \le \frac{7 n}{rd}$ for all $\sfi\in[d]$ such that
\begin{equation*}
    \sum_{\sfi\in[d]} \left(\frac{n_\sfi}{r}\right) Y_\sfi \equiv q \mod \frac{n}{r}~.
\end{equation*}
\end{lemma}
We reproduce the proof of \cref{lem:gcd_result} in \cref{app:gcd_result} for completeness. Let us now prove \cref{lem:int-cube}:
\begin{proof}[Proof of \cref{lem:int-cube}]
For each $a\in[t]$, define
\begin{equation*}
    q_a := \frac{1}{r} \left[m_a - \sum_{\sfk=1}^{d-1} \left(\sum_{\sfi=1}^{\sfk} n_\sfi\right) \xi^*_{\sfk,a}\right]~.
\end{equation*}
Each $q_a$ is an integer because $m_a \equiv 0 \pmod r$ (by hypothesis)  and $n_\sfi \equiv 0 \pmod r$ for all $\sfi\in [d]$. By \cref{lem:gcd_result}, for each $a\in[t]$, there are integers $Y_{1,a},\ldots,Y_{d,a}$, with $|Y_{\sfi,a}| \le \frac{7 n}{rd}$ for all $\sfi\in[d]$, such that
% \begin{equation*}
%     \sum_{\sfk\in[d]} n_\sfk Y_\sfk = r \implies \frac{1}{r}\sum_{\sfk=1}^{d-1} \left(\sum_{\sfi=1}^{\sfk} n_\sfi\right)(Y_\sfk-Y_{\sfk+1}) = \frac{1}{r}\left(\sum_{\sfk=1}^{d-1}  n_\sfk Y_\sfk - \left(n-n_d\right) Y_d \right) \equiv 1 \mod \frac{n}{r}~.
% \end{equation*}
%because $n_d = n-\sum_{\sfi=1}^{d-1}n_\sfi$.
\begin{equation*}
    \sum_{\sfi\in[d]} \left(\frac{n_\sfi}{r}\right) Y_{\sfi,a} \equiv q_a \mod \frac{n}{r}
\end{equation*}
For each $\sfi\in [d-1]$ and $a\in[t]$, choose $\b c^*$ such that 
\begin{equation*}
    c^*_{\sfi,a} = Y_{\sfi,a}-Y_{\sfi+1,a} \in \left[-\frac{14 n}{rd},\frac{14 n}{rd}\right]^{(d-1)t}\cap \mathbb Z^{(d-1)t}~,
\end{equation*}
It follows that
\begin{equation*}
    \frac{1}{r}\sum_{\sfk=1}^{d-1} \left(\sum_{\sfi=1}^{\sfk} n_\sfi\right) c^*_{\sfk,a} \equiv q_a \mod \frac{n}{r}~,
\end{equation*}
for each $a\in[t]$. By hypothesis, $\b \xi^* + \b c^*$ is also an integer point in $\Gamma(\b \sfv,\b m)$.
Moreover, for each $a\in[t]$,
\begin{align*}
    \sfx_{d,a}(\b \xi^* + \b c^*) &= \frac1n \left[m_a - \sum_{\sfk=1}^{d-1} \left(\sum_{\sfi=1}^{\sfk} n_\sfi\right) \xi^*_{\sfk,a} - \sum_{\sfk=1}^{d-1} \left(\sum_{\sfi=1}^{\sfk} n_\sfi\right) c^*_{\sfk,a} \right]
    \\
    &= \frac{r}n \left[q_a - \frac{1}{r}\sum_{\sfk=1}^{d-1} \left(\sum_{\sfi=1}^{\sfk} n_\sfi\right) c^*_{\sfk,a} \right]\in\mathbb Z~.
\end{align*}
So, $\b \xi = \b \xi^*+\b c^*$ together with $\sfx_{d,a}(\b \xi)$ maps to an integer point in $\Delta(\b v,\b m)$.
\end{proof}

Thus, we just have to show that $\Gamma(\b \sfv, \b m)$ contains a $(d-1)t$-dimensional ``integer cube'' of length $2\lfloor\frac{14 n}{rd}\rfloor$, with edges parallel to the coordinate axes.
For this, it suffices to show that $\Gamma(\b \sfv, \b m)$ contains a $(d-1)t$-dimensional cube of length $\frac{28 n}{rd}+1$, with edges parallel to the coordinate axes, because any (closed) interval of length $2\gamma+1$ on the real line, for some $\gamma\ge0$ contains $2\lfloor\gamma\rfloor+1$ consecutive integers.
For example, in $\mathbb R^2$, a square of length $2\gamma+1$ with edges parallel to the coordinate axes, say $[\alpha,\alpha+2\gamma+1]\times[\beta,\beta+2\gamma+1]$ for some $\alpha,\beta\in \mathbb R$, contains the ``integer square'' of length $2\lfloor\gamma\rfloor$ with corners $(\lceil \alpha\rceil,\lceil \beta\rceil)$, $(\lceil \alpha\rceil+2\lfloor\gamma\rfloor,\lceil \beta\rceil)$, $(\lceil \alpha\rceil,\lceil \beta\rceil+2\lfloor\gamma\rfloor)$, and $(\lceil \alpha\rceil+2\lfloor\gamma\rfloor,\lceil \beta\rceil+2\lfloor\gamma\rfloor)$.

In what follows, we show that $\Gamma(\b \sfv,\b m)$ contains a $(d-1)t$-dimensional cube of length $\frac{28 n}{rd}+1$, with edges parallel to the coordinate axes, when $m_a$'s are large enough.
Let us relax $m_a$ to be real.
Since $\Gamma(\b \sfv,\b m)$ is an intersection of \emph{half-spaces}, it is a convex region in $\mathbb R^{(d-1)t}$.
Moreover, the following is true:
\begin{lemma}\label{lem:volume}
If the valuations $\sfV_1,\ldots,\sfV_d$ are additive and pairwise distinct, then $\Gamma(\b \sfv,\b m)$ has nonzero $(d-1)t$-dimensional volume for any $\b m \in \mathbb R^t_+$.
\end{lemma}
We prove this lemma soon but pursue its consequences now.
In particular, one consequence of \cref{lem:volume}, and convexity, is that $\Gamma(\b \sfv, \b 1^t)$ contains a $(d-1)t$-dimensional cube of nonzero length, say $\epsilon>0$ (which depends on $\b \sfv$), with edges parallel to the coordinate axes.
What can we say about $\Gamma(\b \sfv, \mu \b 1^t)$ for $\mu>0$?

\begin{lemma}\label{lem:scale}
For any $\lambda>0$, $\b \xi \in \Gamma(\b \sfv,\b m)$ if and only if $\lambda \b \xi \in \Gamma(\b \sfv,\lambda \b m)$.
In other words, $\Gamma(\b \sfv,\lambda \b m) = \lambda \Gamma(\b \sfv,\b m):=\{\lambda \b \xi:\b \xi \in \Gamma(\b \sfv,\b m)\}$.
\end{lemma}
\begin{proof}
Note that \eqref{EF} is homogeneous in $\b \xi$, i.e., $\b \xi$ satisfies \eqref{EF} if and only if $\lambda \b \xi$ satisfies \eqref{EF}.
On the other hand, $\b \xi$ satisfies \eqref{pos} with parameter $\b m$ if and only if $\lambda \b \xi$ satisfies \eqref{pos} with parameter $\lambda \b m$.
\end{proof}
It follows from \cref{lem:scale} that $\Gamma(\b \sfv, \mu\b 1^t)$ contains a $(d-1)t$-dimensional cube of length $\mu\epsilon$ with edges parallel to the cordinate axes.
If we choose $\mu$ to be a positive integer such that
\begin{equation}
    \mu>\left(\frac{28 n}{rd}+1\right) \frac{1}{\epsilon}~,\label{eq:mu}
\end{equation}
then, as we desire, $\Gamma(\b \sfv, \mu\b 1^t)$ contains a $(d-1)t$-dimensional cube of length $\frac{28 n}{rd}+1$ with edges parallel to the coordinate axes!

Now comes the punchline: unlike $\Delta(\b v,\b m)$, $\Gamma(\b \sfv,\b m)$ has a nice ``containment'' property.
\begin{lemma}\label{lem:contain}
If $\b m'\ge \b m$, then $\Gamma(\b \sfv,\b m') \supseteq \Gamma(\b \sfv,\b m)$.
\end{lemma}
\begin{proof}
Say $\b \xi\in \Gamma(\b \sfv,\b m)$.
Then, $\b \xi$ satisfies \eqref{EF} and \eqref{pos} with parameters $\b \sfv$ and $\b m$.
Since $\b m'\ge \b m$, we have 
\begin{align*}
    &m_a' - \sum_{\sfk=1}^{\sfi-1} \left(\sum_{\sfj=1}^{\sfk} n_\sfj\right) \xi_{\sfk,a} + \sum_{\sfk=\sfi}^{d-1} \left(n-\sum_{\sfj=1}^{\sfk} n_\sfj\right) \xi_{\sfk,a} 
    \\
    &\ge m_a - \sum_{\sfk=1}^{\sfi-1} \left(\sum_{\sfj=1}^{\sfk} n_\sfj\right) \xi_{\sfk,a} + \sum_{\sfk=\sfi}^{d-1} \left(n-\sum_{\sfj=1}^{\sfk} n_\sfj\right) \xi_{\sfk,a} \ge 0~,\quad \forall \sfi\in[d]~,~a\in[t]~.
\end{align*}
Therefore, $\b \xi$ satisfies \eqref{EF} and \eqref{pos} with parameters $\b \sfv$ and $\b m'$.
In other words, $\b \xi\in \Gamma(\b \sfv,\b m')$.
\end{proof}

\cref{lem:contain} ensures that $\Gamma(\b \sfv,\b m)$ contains a $(d-1)t$-dimensional cube of length $\frac{28 n}{rd}+1$ with edges parallel to the coordinate axes, whenever $\b m\ge \mu\b 1^t$.
Thus, \cref{thm:EF} is proved.

Let us tie up the only remaining loose end: proving \cref{lem:volume}.
\begin{proof}[Proof of \cref{lem:volume}]
Let us replace \eqref{EF} with ``stronger'' inequalities
\begin{equation*}
    \sum_{a\in[t]} \sfv_\sfi(a) \xi_{\sfk,a} \ge 0~,\quad\text{and}\quad \sum_{a\in[t]} \sfv_\sfj(a) \xi_{\sfk,a} \le 0~, \quad \forall 1\le \sfi\le \sfk<\sfj\le d~.
\end{equation*}
After removing repetitions, these are same as the inequalities
\begin{equation}
    \begin{aligned}
        &\sum_{a\in[t]} \sfv_\sfi(a) \xi_{\sfj,a} \ge 0~,\quad\forall 1\le \sfi\le \sfj< d~,
        \\
        &\sum_{a\in[t]} \sfv_\sfi(a) \xi_{\sfj,a} \le 0~,\quad\forall 1\le \sfj<\sfi\le d~.
    \end{aligned}\tag{\ref*{EF}${}'$} \label{EF-new}
\end{equation}
Let $\Gamma'(\b \sfv,\b m)\subseteq \mathbb R^{(d-1)t}$ denote the feasible set of the modified constraints \eqref{EF-new} and \eqref{pos}.
It is clear that $\Gamma'(\b \sfv,\b m) \subseteq \Gamma(\b \sfv,\b m)$ because \eqref{EF-new} is stronger than \eqref{EF}, i.e., if $\b \xi$ satisfies \eqref{EF-new} then it satisfies \eqref{EF} as well.
So, it suffices to show that the $(d-1)t$-dimensional volume of $\Gamma'(\b \sfv,\b m)$ is nonzero.

We must handle one subtlety in replacing \eqref{EF} with \eqref{EF-new}.
Although the LP problem given by \eqref{EF} and \eqref{pos} is independent of the ordering on the valuations $\sfV_1,\ldots,\sfV_d$, the modified LP problem given by \eqref{EF-new} and \eqref{pos} does depend on this ordering.
We show that there is an ordering (not necessarily unique) on the valuations for which the $(d-1)t$-dimensional volume of $\Gamma'(\b \sfv,\b m)$ is nonzero.

Let us first find a convenient ordering on the valuations.
For each $\sfi\in[d]$, consider the hyperplane in $\mathbb R^t$, passing through the origin $\b 0^t$, given by $H_\sfi := \{\b z\in \mathbb R^t : \sum_{a\in[t]} \sfv_\sfi(a) z_a = 0\}$.
We denote the \emph{half-spaces} of the hyperplane $H_\sfi$ by $H_\sfi^+$ and $H_\sfi^-$ respectively.
Namely, 
\[H_\sfi^+ = \{\b z\in \mathbb R^t : \textstyle \sum_{a\in[t]} \sfv_\sfi(a) z_a \ge 0\}~, \qquad H_\sfi^- = \{\b z\in \mathbb R^t : \textstyle \sum_{a\in[t]} \sfv_\sfi(a) z_a \le 0\}~.\]
Note that $\b \xi$ satisfies \eqref{EF-new} if and only if $(\xi_{\sfj,1},\ldots,\xi_{\sfj,t}) \in \Sigma^{(j)}(\b \sfv) := H_1^+ \cap \cdots \cap H_\sfj^+ \cap H_{\sfj+1}^- \cap \cdots \cap H_d^-\subseteq \mathbb R^t$ for each $\sfj\in[d-1]$.
In other words, $\b \xi$ satisfies \eqref{EF-new} if and only if $\b \xi \in \Sigma(\b \sfv):=\Sigma^{(1)}(\b \sfv) \times \cdots \times \Sigma^{(d-1)}(\b \sfv)\subseteq \mathbb R^{(d-1)t}$.

Since the valuations are assumed to be distinct, the hyperplanes are all distinct.
Indeed, say two hyperplanes $H_\sfi$ and $H_\sfj$ are the same. Then there is a $\lambda>0$ such that $\sfv_\sfi(a)=\lambda \sfv_\sfj(a)$ for all $a\in[t]$. This implies that $\sfV_\sfi$ and $\sfV_\sfj$ are identical valuations contradicting our assumption.

Since the hyperplanes are distinct, there is a directed line $L$ in $\mathbb R^t$ that does not pass through the origin $\b 0^t$, but passes through the \emph{positive} and \emph{negative orthants}, $\mathbb R^t_+$ and $\mathbb R^t_-$, and intersects the $d$ hyperplanes \emph{transversally} in $d$ \emph{distinct} points, $\b p_1,\ldots,\b p_d$ \cite{hyperplane}.\footnote{Such a line always exists when $t\ge 2$ but not when $t=1$.
This is okay because we do not consider $t=1$ anyway.
Recall that by \cref{def:identical-distinct}, there are no distinct valuations when $t=1$.}
Let $\b z_+$ and $\b z_-$ be two points on $L$ such that $\b z_+\in\mathbb R^t_+$ and $\b z_- \in \mathbb R^t_-$.
The direction of $L$ is given by an arrow pointing from $\b z_-$ to $\b z_+$.
Let the valuations be ordered such that, as we traverse the line $L$ along this arrow, we first cross $H_1$ at $\b p_1$, followed by $H_2$ at $\b p_2$, and so on.
Note that this ordering depends on the choice of $L$ but our proof is independent of this choice.
See \cref{fig:hyperplane-order} for an illustration of this ordering when $t=2$ and $d=3$.

\begin{figure}
    \centering
    \includegraphics[scale=0.25]{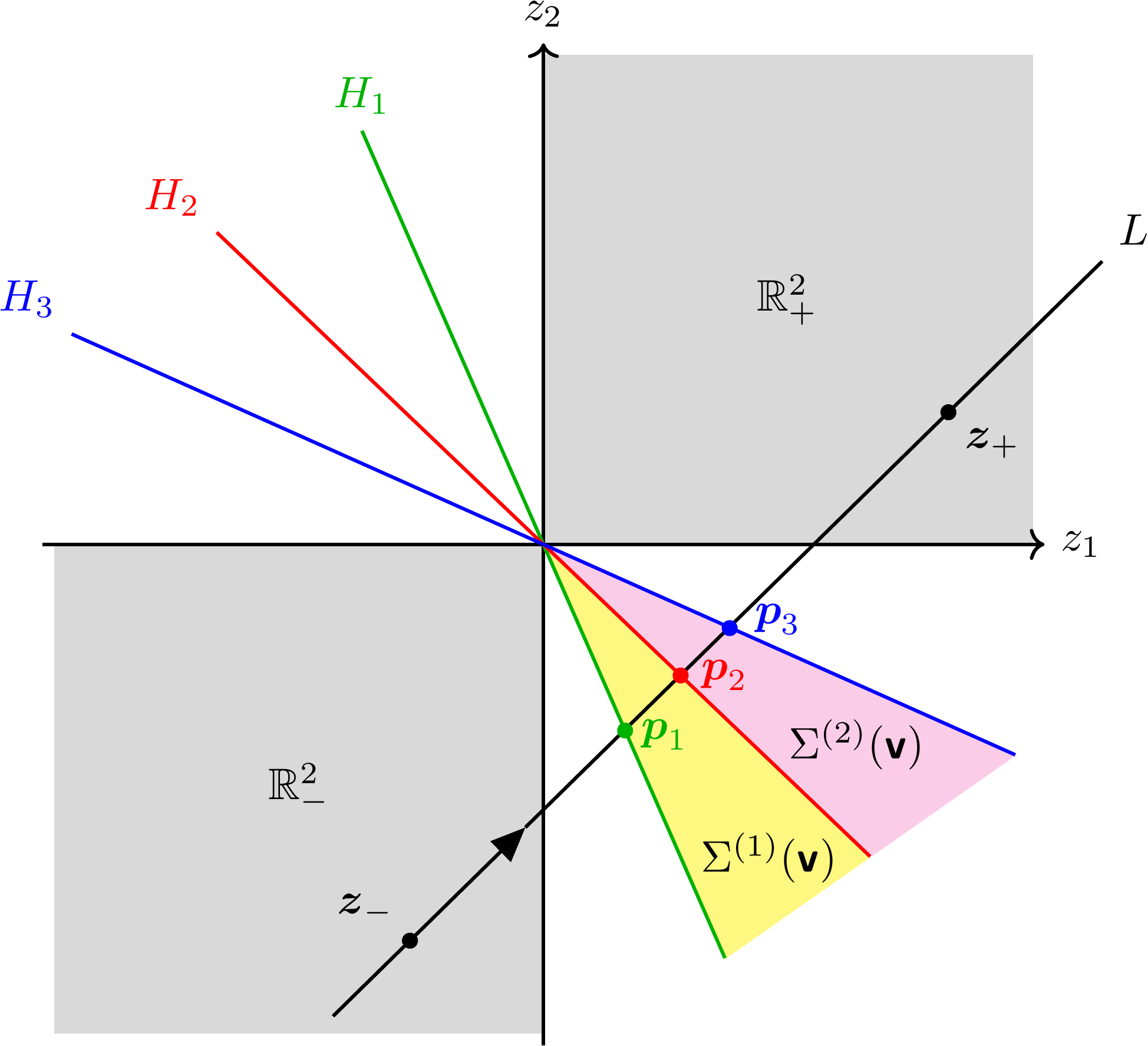}
    \caption{\small Illustration of how to order the agents when $t=2$ and $d=3$ in the proof of \cref{lem:volume}.
    The hyperplanes associated with the valuations $\sfV_1,\sfV_2,\sfV_3$ are the lines $H_1,H_2,H_3$.
    The directed line $L$ does not pass through the origin, but passes through the positive and negative quadrants (orthants in two dimensions; gray regions), and intersects the hyperplanes transversally in points $\b p_1,\b p_2,\b p_3$.
    The valuations are ordered such that as we traverse $L$ along the arrow (from $\b z_-$ to $\b z_+$), we cross the hyperplanes in the order $H_1,H_2,H_3$.
    This order ensures that the regions $\Sigma^{(1)}(\b \sfv) = H_1^+ \cap H_2^- \cap H_3^-$ and $\Sigma^{(2)}(\b \sfv) = H_1^+ \cap H_2^+ \cap H_3^-$, shown in yellow and pink, have nonzero area (2-dimensional volume).}
    \label{fig:hyperplane-order}
\end{figure}

With this ordering on the valuations, for each $\sfj\in[d-1]$, the segment of $L$ between the points $\b p_\sfj$ and $\b p_{\sfj+1}$ is contained in the region $\Sigma^{(\sfj)}(\b \sfv)$.
By the choice of $L$, any \emph{tubular neighborhood}\footnote{Intuitively, a tubular neighborhood of a line $\ell$ in $\mathbb R^t$ can be thought of as ``thickening'' $\ell$ to a $t$-dimensional cylinder, or a \emph{tube}, containing it.} of this segment intersects $\Sigma^{(\sfj)}(\b \sfv)$ in a region of nonzero $t$-dimensional volume.
So, $\Sigma^{(\sfj)}(\b \sfv)$ itself has nonzero $t$-dimensional volume (see \cref{fig:hyperplane-order}), and hence, $\Sigma(\b \sfv)$ has nonzero $(d-1)t$-dimensional volume.

Having taken care of \eqref{EF-new}, let us now turn to \eqref{pos}.
Consider the $(d-1)t$-dimensional cube $[-\rho,\rho]^{(d-1)t}$, where
\begin{equation}
    \rho := \frac{\min_{a\in[t]} m_a}{n(d-1)}~.\label{eq:muho}
\end{equation}
Note that $\rho>0$ because $\b m\in\mathbb R_+^t$ by hypothesis, so $[-\rho,\rho]^{(d-1)t}$ has nonzero $(d-1)t$-dimensional volume.
Then, any $\b \xi\in [-\rho,\rho]^{(d-1)t}$ satisfies \eqref{pos} because for any $\sfi\in[d]$ and $a\in[t]$,
\begin{align*}
    &m_a - \sum_{\sfk=1}^{\sfi-1} \left(\sum_{\sfj=1}^{\sfk} n_\sfj\right) \xi_{\sfk,a} + \sum_{\sfk=\sfi}^{d-1} \left(n-\sum_{\sfj=1}^{\sfk} n_\sfj\right) \xi_{\sfk,a}
    \\
    &\ge m_a - \sum_{\sfk=1}^{\sfi-1} \left(\sum_{\sfj=1}^{\sfk} n_\sfj\right) \rho - \sum_{\sfk=\sfi}^{d-1} \left(n-\sum_{\sfj=1}^{\sfk} n_\sfj\right) \rho
    \\
    &\ge m_a - n(d-1) \rho \ge 0~.
\end{align*}
Since any $\b \xi \in \Sigma(\b \sfv) \cap [-\rho,\rho]^{(d-1)t}$ satisfies both \eqref{EF-new} and \eqref{pos}, it follows that $\Gamma'(\b \sfv,\b m)$ contains $\Sigma(\b \sfv) \cap [-\rho,\rho]^{(d-1)t}$.
Since $\Sigma(\b \sfv)$ is a convex region with the origin $\b 0^{(d-1)t}$ as a vertex (because it is an intersection of half-spaces of hyperplanes passing through the origin), the intersection $\Sigma(\b \sfv) \cap [-\rho,\rho]^{(d-1)t}$ has nonzero $(d-1)t$-dimensional volume.
Therefore, $\Gamma'(\b \sfv,\b m)$ has nonzero $(d-1)t$-dimensional volume as well.
\end{proof}

\subsection{Bounds on $\mu$ in \cref{thm:EF}}\label{sec:bound_mu}

\subsubsection{Upper bound on $\mu$ for $d=2$}\label{sec:upperbound_mu}
In this subsection, we show that for $d=2$, $\mu = O(n^2 \sqrt t/(r\delta))$ using the techniques developed in the previous subsection.
For each $\sfi\in[d]$, the hyperplane $H_\sfi$ passing through the origin is uniquely determined by the normal vector $\b \sfv_\sfi := (\sfv_\sfi(1),\ldots,\sfv_\sfi(t))\in\mathbb R^t_{\ge0}$.
Let $\theta_{\sfi\sfj}$ be the angle between $\b \sfv_\sfi$ and $\b \sfv_\sfj$ given by
\begin{equation*}
    \cos \theta_{\sfi\sfj} = \frac{\b \sfv_\sfi \cdot \b \sfv_\sfj}{\Vert \b \sfv_\sfi\Vert_2 \Vert \b \sfv_\sfj\Vert_2}~.
\end{equation*}
Since the valuations are distinct, we have $0 < \theta_{\sfi\sfj} \le \frac{\pi}{2}$.
For any $0<\delta\le\frac{\pi}{2}$, we say that the valuations $\sfV_\sfi$ and $\sfV_\sfj$ are ``\emph{$\delta$-far from being identical}'' if $\theta_{\sfi\sfj}\ge\delta$.

\begin{theorem}\label{thm:upperboundmu}
If $d=2$ and the valuations $\sfV_1$ and $\sfV_2$ are $\delta$-far from being identical, then $\mu = O(n^2 \sqrt t/(r\delta))$.
\end{theorem}

\begin{proof}
Recall that, by \eqref{eq:mu}, we can choose $\mu = \lceil\left(\frac{14 n}{r}+1\right)\frac{1}{\epsilon}\rceil$, where $\epsilon$ is the length of the largest $t$-dimensional cube, with edges parallel to the coordinate axes, inside $\Gamma(\b \sfv,\b 1^t)$.
A lower bound on $\epsilon$ immediately gives an upper bound on $\mu$.

With $\b m = \b 1^t$ in \eqref{eq:muho}, we have $\rho = 1/n$.
Consider the intersection $\Sigma(\b \sfv) \cap [-\rho,\rho]^t$.
We claim that we can always fit inside this intersection a $t$-dimensional cube, with edges parallel to the coordinate axes, of length
\begin{equation*}
    \epsilon' = \frac{2\rho}{\sqrt{t}}\cdot\frac{\sin(\delta/2)}{1+\sin(\delta/2)} = \frac{1}{\sqrt{t}}\cdot\frac{2\sin(\delta/2)}{1+\sin(\delta/2)} \cdot \frac{1}{n}~.
\end{equation*}
Let us call this cube $C(\b \sfv)$.
%Consider the $(d-1)t$-dimensional cube $C(\b \sfv) := C^{(1)}(\b \sfv)\times \cdots \times C^{(d-1)}(\b \sfv)$. Since each $C^{(\sfj)}(\b \sfv) \subseteq \Sigma^{(\sfj)}(\b \sfv) \cap [-\rho,\rho]^t$, we conclude that $C(\b \sfv)\subseteq \Sigma(\b \sfv) \cap [-\rho,\rho]^{(d-1)t} \subseteq \Gamma(\b \sfv, \b 1^t)$.
Thus, $\epsilon \ge \epsilon'$, and therefore,
\begin{equation*}
    \mu \le \frac{1+\sin(\delta/2)}{2\sin(\delta/2)} \cdot \sqrt{t} \cdot n \cdot \left(\frac{14 n}{r}+1\right) + 1 = O\left( \frac{n^2\sqrt{t}}{r\delta}\right)~.
\end{equation*}

To finish the proof, we have to prove the above claim on existence of cube of length $\epsilon'$.
Let $B(\b 0^t,\rho)$ be the $t$-dimensional ball of radius $\rho$ centered at the origin.
Clearly, $B(\b 0^t,\rho) \subseteq [-\rho,\rho]^t$.
Consider the intersection $W(\b \sfv):=\Sigma(\b \sfv) \cap B(\b 0^t,\rho)$.
Our approach is to show that we can always fit a $t$-dimensional ball of radius $\rho':=\sqrt{t}\epsilon'/2$ inside $W(\b \sfv)$, because we can then fit a $t$-dimensional cube $C(\b \sfv)$ of length $\epsilon'$ inside this ball.

%\gnote{the following sentence is incorrect; this might change the upper bound drastically} 
The region $W(\b \sfv)$ is bounded by the hyperplanes $H_1$ and $H_2$ passing through the origin, and the sphere of radius $\rho$ centered at the origin.
Such a region is called a \emph{spherical wedge}.
When $t=2$, it is more commonly called a \emph{sector} or a \emph{pie}.

\begin{figure}
    \centering
    \includegraphics[scale=0.23]{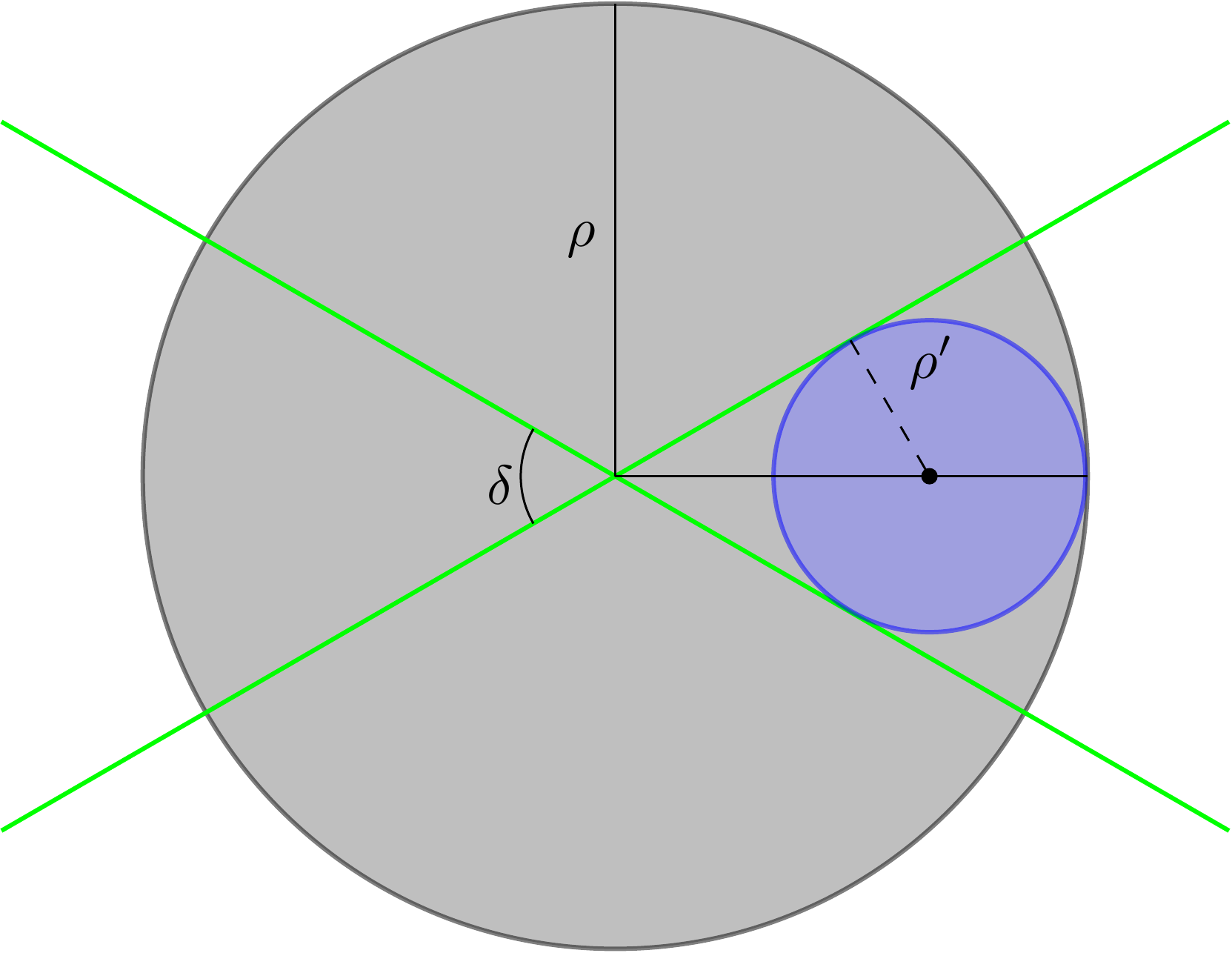}~~~~~~~~~~~~~~~\includegraphics[scale=0.35]{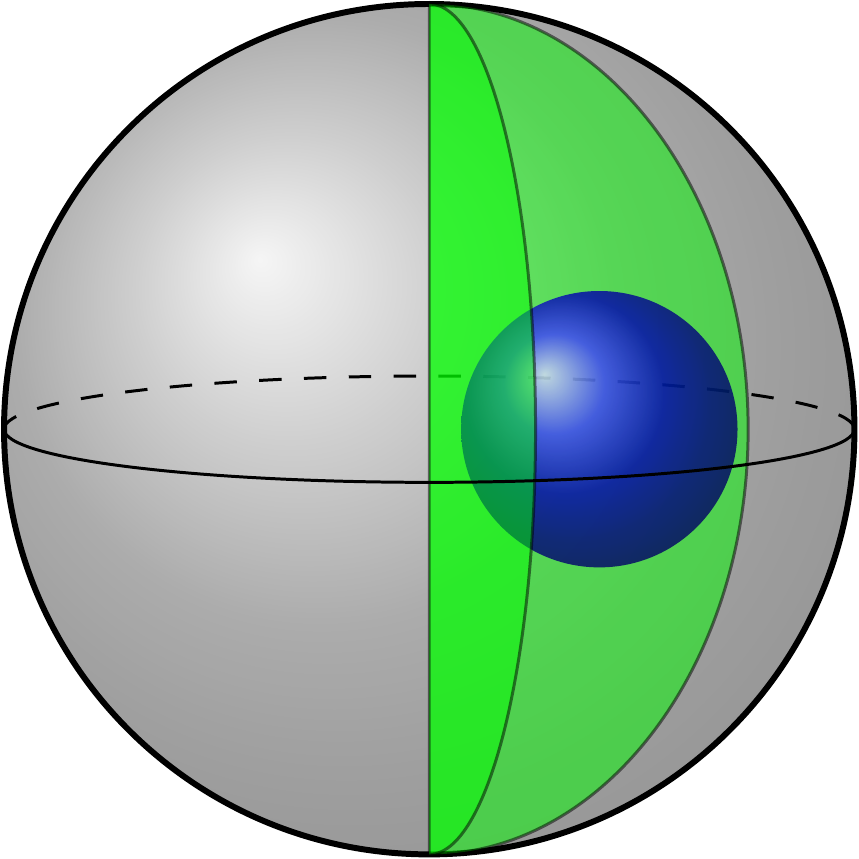}\\
    ~~~~~(a) $t=2$~~~~~~~~~~~~~~~~~~~~~~~~~~~~~~~~~~~~~~~~~~~~~~~~~(b) $t=3$
    \caption{\footnotesize Fitting the largest ball inside a spherical wedge when $t=2$ and $t=3$. The green objects are the two hyperplanes passing through the origin with the included angle $\delta$, the gray objects are the balls of radius $\rho$ centered at the origin, and the blue objects are the largest balls of radius $\rho'$ that can fit inside the spherical wedge.}
    \label{fig:spherical-wedge}
\end{figure}

When $t=2$, a simple geometric construction, as shown in \cref{fig:spherical-wedge}(a), shows that the largest disc (blue disc) that can fit inside a sector of angle $\delta$ has radius
\begin{equation*}
    \rho' = \rho\cdot \frac{\sin(\delta/2)}{1+\sin(\delta/2)} = \frac{\sqrt{t} \epsilon'}{2}~.
\end{equation*}
When $t=3$, as shown in \cref{fig:spherical-wedge}(b), the (blue) ball inside the spherical wedge is largest when the horizontal equatorial planes of the two balls coincide.
So, the problem again reduces to a $2$-dimensional problem in the common equatorial plane.
In general, for any $t$, the problem always reduces to a $2$-dimensional problem, and hence the radius of the largest ball inside the spherical wedge of angle $\delta$ is given by $\rho'$ above.

Since $\rho'$ is a monotonically increasing function of $\delta$ for $0<\delta\le \frac{\pi}{2}$, and since the valuations are $\delta$-far from being identical, a $t$-dimensional ball of radius $\rho'$ fits inside $W(\b \sfv)$.
\end{proof}

\subsubsection{Upper bound on $\mu$ for $t=2$}\label{sec:upperbound_mu2}

\begin{theorem}\label{thm:upperboundmu2}
If $t=2$ and the valuations $\sfV_1,\dots,\sfV_d$ are pairwise $\delta$-far from being identical, then $\mu = O(n^2 /(r\delta))$.
\end{theorem}
Note that, while $\delta$ is an independent parameter, its range depends on $d$ and $t$, i.e., $0<\delta\le \delta^*(d,t)$, where $\delta^*(d,t)$ is the maximum possible \emph{min-angle}\footnote{Here, min-angle is the minimum of the angles between all pairs of vectors.} between $d$ unit vectors in $\mathbb R^t_{\ge0}$.
For example, when $t=2$, the maximum possible min-angle between $d$ unit vectors in a quadrant is $\delta^*(d,2) = \frac{\pi}{2(d-1)}$.
When $t=3$, this is related to the \emph{Tammes' problem} \cite{Tammes:1930}, and the solution is known only for a finite number of $d$'s.

\begin{proof}[Proof of \cref{thm:upperboundmu2}]
This proof is very similar to the proof of \cref{thm:upperboundmu}.
Again, by \eqref{eq:mu}, we can choose $\mu = \lceil\left(\frac{28 n}{rd}+1\right)\frac{1}{\epsilon}\rceil$, where $\epsilon$ is the length of the largest $2$-dimensional cube, i.e., a square, with edges parallel to the coordinate axes, inside $\Gamma(\b \sfv,\b 1^2)$.
A lower bound on $\epsilon$ immediately gives an upper bound on $\mu$.

With $\b m = \b 1^2$ in \eqref{eq:muho}, we have $\rho = 1/n(d-1)$.
Consider the intersection $\Sigma^{(\sfj)}(\b \sfv) \cap [-\rho,\rho]^2$ for $\sfj\in [d-1]$.
We claim that we can always fit inside this intersection a square, with edges parallel to the coordinate axes, of length
\begin{equation*}
    \epsilon' = \frac{2\rho}{\sqrt{2}}\cdot\frac{\sin(\delta/2)}{1+\sin(\delta/2)} = \frac{1}{\sqrt{2}}\cdot\frac{2\sin(\delta/2)}{1+\sin(\delta/2)} \cdot \frac{1}{n(d-1)}~.
\end{equation*}
Let us call this square $C^{(\sfj)}(\b \sfv)$.
Consider the $2(d-1)$-dimensional cube $C(\b \sfv) := C^{(1)}(\b \sfv)\times \cdots \times C^{(d-1)}(\b \sfv)$. Since each $C^{(\sfj)}(\b \sfv) \subseteq \Sigma^{(\sfj)}(\b \sfv) \cap [-\rho,\rho]^2$, we conclude that $C(\b \sfv)\subseteq \Sigma(\b \sfv) \cap [-\rho,\rho]^{2(d-1)} \subseteq \Gamma(\b \sfv, \b 1^2)$.
Thus, $\epsilon \ge \epsilon'$, and therefore,
\begin{equation*}
    \mu \le \frac{1+\sin(\delta/2)}{2\sin(\delta/2)} \cdot \sqrt{2} \cdot n (d-1) \cdot \left(\frac{28 n}{rd}+1\right) + 1 = O\left( \frac{n^2}{r\delta}\right)~.
\end{equation*}

To finish the proof, we have to prove the above claim on existence of square of length $\epsilon'$.
Consider the intersection $W^{(\sfj)}(\b \sfv):=\Sigma^{(\sfj)}(\b \sfv) \cap B(\b 0^2,\rho)$.
We show that we can always fit a disc of radius $\rho':=\epsilon'/\sqrt{2}$ inside $W^{(\sfj)}(\b \sfv)$, because we can then fit a square $C^{(\sfj)}(\b \sfv)$ of length $\epsilon'$ inside this ball.

The region $W^{(\sfj)}(\b \sfv)$ is bounded by the lines $H_\sfj$ and $H_{\sfj+1}$ passing through the origin, and the disc of radius $\rho$ centered at the origin. Like in the proof of \cref{thm:upperboundmu}, the largest disc that can fit inside a sector of angle $\delta$ has radius
\begin{equation*}
    \rho' = \rho\cdot \frac{\sin(\delta/2)}{1+\sin(\delta/2)} = \frac{\epsilon'}{\sqrt{2}}~.
\end{equation*}

Since $\rho'$ is a monotonically increasing function of $\delta$ for $0<\delta\le \frac{\pi}{2}$, and since the valuations are $\delta$-far from being identical, a disc of radius $\rho'$ fits inside $W^{(\sfj)}(\b \sfv)$.
\end{proof}

%\vnote{Pranay TODO: Add a remark about why we can't do this for arbitrary $t$ and $d$. Emphasize that even $t=d=3$ is hard.}

\begin{remark*}
Note that the above techniques are not helpful in determining the upper bound on $\mu$ when $t>2$ and $d>2$. We demonstrate the issues that arise when $t=d=3$ (similar issues persist for larger values of $t$ and $d$). Consider three item-value functions given by $\sfv_\sfi(1) = 1$, $\sfv_\sfi(2) = 1 + 2\sfi \delta$, for $\sfi\in[3]$, and $\sfv_1(3) = \sfv_2(3) = 0$, $\sfv_3(3) = \delta' $, where $0<\delta',\delta\ll1$. For sufficiently small $\delta$, the valuations $\sfV_1,\sfV_2,\sfV_3$ are $\delta$-far from each other. When $\delta'=0$, we get an effectively $t=2$ problem, for which we know the upper bound on $\mu$ from \cref{sec:upperbound_mu2}. However, when $0<\delta'\ll\delta$, there is a thin region bounded by the three planes $H_1,H_2,H_3$. If the line $L$, chosen to order the valuations in the proof of \cref{lem:volume}, passes through this region, then one of the $\Sigma^{(\sfj)}(\b \sfv)$'s is this thin region. Moreover, the length $\epsilon'$ of the largest cube that one can fit inside this thin region is arbitrarily small, so the upper bound on $\mu$ is arbitrarily large. One can get a better upper bound on $\mu$ by choosing the line $L$ so as to avoid such thin regions. Even in the case of $t=d=3$, where there are three planes in $\mathbb R^3$, finding such an optimal line $L$ is hard. This makes the techniques of previous sections not helpful in finding an upper bound on $\mu$ when $t>2$ and $d>2$.
\end{remark*}

\subsubsection{Lower bound on $\mu$ for $d=2$}\label{sec:lowerbound_mu}

In this subsection, we show that for $d=2$, $\mu = \Omega(n^2/(r\delta))$. In particular, we show the following: %\gnote{This lower bound could be smaller for $d>2$. We can show that $\mu = \Omega(n^2/(d^3r\delta))$ for $d=o(n)$, and $\mu = \Omega(1/\delta)$ for $d=n$. We can also show that $\mu = \Omega(n^2/(d^4r\delta))$ for any $d$.}
%\vnote{Note sure if it makes sense to include this lower bound when we do not have any upper bound for arbitrary $d$. We also don't have a single unified lower bound.}
%\knote{I agree with S, I think it's fine to leave the bound for $d=2$.}
\begin{theorem}
\label{thm:lowerboundmu}
For every \gnote{added}$t\ge2$, $n\ge 2$, positive proper divisor $r$ of $n$, and sufficiently small $\delta>0$, there exists an instance with
\begin{itemize}
    \item $t$ types of items,
    \item $d=2$ distinct additive valuations $\sfV_1$ and $\sfV_2$ that are $\delta$-far from being identical,
    \item $n_1$ agents with valuation $\sfV_1$ and $n_2$ agents with valuation $\sfV_2$ such that $\gcd(n_1,n_2)=r$ and $n_1+n_2=n$, and
    \item at least $n^2/(48r\delta)$ items of each type
\end{itemize}
such that there is no complete EF allocation.
\end{theorem}

\begin{proof}
Given an $n\ge 2$, we choose $n_1=n-r$ and $n_2=r$ so that $n_1+n_2=n$ and $\gcd(n_1,n_2)=r$. Since $r$ is a positive proper divisor of $n$, we have $r\le n/2$, which means $n_1 \ge n/2$.

We first prove the theorem for $t=2$. We choose the item-value functions to be given by $\sfv_1(1) =\sfv_2(1) = 1$, $\sfv_1(2) = 1+4\delta'$, and $\sfv_2(2) = 1+\delta'$, where $\delta<\delta'<2\delta$. Then, the angle between the vectors $\b \sfv_1 = (1,1+4\delta')$ and $\b \sfv_2 = (1,1+\delta')$ is at least $\delta$, i.e., $\theta_{12} > \delta$ for $0<\delta<\frac18$, so $\sfV_1$ and $\sfV_2$ are $\delta$-far from being identical. Moreover, since $\delta>0$, we can choose $\delta'$ to be an irrational number between $\delta$ and $2\delta$.

For any $\sfi\in[2]$, if $i,j\in N_\sfi$, the EF inequalities \eqref{EF-def} between agents $i$ and $j$ become equations because they have identical valuations. That is,
\begin{equation*}
    \sum_{a\in[2]} \sfv_\sfi(a) (x_{i,a}-x_{j,a}) = 0~,
\end{equation*}
where $x_{i,a}$ is the number of type-$a$ items allocated to agent $i$. Since $\sfv_\sfi(1)/\sfv_\sfi(2)$ is positive and irrational by construction, the only integral solution to the above equation is $x_{i,a}=x_{j,a}$ for all $a\in[2]$. In other words, all agents with identical valuations must receive the same set of items.

Let $\mu_0 = (n^2/r)\lceil 1/(24\delta')\rceil\ge (n^2/r)\lceil 1/(48\delta)\rceil$, and let $M$ be a multiset with
\begin{equation*}
    m_1 = \mu_0 + n - r~,\qquad m_2 = \mu_0~.
\end{equation*}
Note that $m_a \equiv 0 \pmod r$ for $a\in[2]$.
Let us define $\xi_a := \sfx_{1,a} - \sfx_{2,a}$, where $\sfx_{\sfi,a}$ is the number of type-$a$ items allocated to each agent in $N_\sfi$.
Inverting this definition using the completeness constraint \eqref{comp-def}, we find that
\begin{align*}
    \sfx_{1,a} = \frac1n (m_a + n_2\xi_a)~, \qquad \sfx_{2,a} = \frac1n (m_a - n_1\xi_a)~.
\end{align*}
Since $\frac{m_1}{r} \equiv (\frac{n}{r}-1) \mod \frac{n}{r}$ and $\frac{m_2}{r}\equiv 0 \mod \frac{n}{r}$, the integers $\xi_a$ must satisfy
\begin{equation}
    \frac{n_1}{r} \xi_1 \equiv \left(\frac{n}{r}-1\right) \mod \frac{n}{r}~,\qquad \frac{n_1}{r} \xi_2 \equiv 0 \mod \frac{n}{r}~,\label{congruence}
\end{equation}
so that $\sfx_{\sfi,a}$'s are integers.

The remaining EF and positivity constraints \eqref{EF-def} and \eqref{pos-def} give
\begin{align*}
    &\xi_1 + (1+4\delta') \xi_2 \ge 0~, &&-\frac{m_a}{n_2} \le \xi_a \le \frac{m_a}{n_1}~,\quad\forall a\in[2]~.
    \\
    &\xi_1 + (1+\delta') \xi_2 \le 0~,
\end{align*}
From the EF constraints on the left, we have $\xi_2 \ge 0$ and $\xi_1 \le 0$ because
\begin{equation*}
    (1+4\delta')\xi_2 \ge -\xi_1 \ge (1+\delta')\xi_2 \implies 3\delta' \xi_2 \ge 0~,
\end{equation*}
and $\delta'>0$.
Moreover, we have
\begin{align*}
    -\xi_1 &\le (1+4\delta') \xi_2 &-\xi_1 &\ge (1+\delta') \xi_2
    \\
    &\le \xi_2 + \frac{4\delta' m_2}{n_1} &&\ge \xi_2~.
    \\
    &\le \xi_2 + \frac{8\delta' \mu_0}{n}
    \\
    &\le \xi_2 + \frac{n}{r}\left(\frac13 + 8\delta'\right)~,
\end{align*}
So, for $\delta'<1/48$, we have 
\begin{equation*}
    \xi_2\le -\xi_1 < \xi_2 + \frac{n}{2r}~,
\end{equation*}
Since $\xi_1$ and $\xi_2$ are integers, $\xi_1 = -\xi_2 - k$ for an integer $0\le k<n/(2r)$. These solutions do not satisfy \eqref{congruence} because
\begin{equation*}
    \frac{n_1}{r}(\xi_1 + \xi_2) = -\frac{k n_1}{r} \equiv k\mod \frac{n}{r}~,\quad \text{but}\quad k\not\equiv \left(\frac{n}{r}-1\right) \mod \frac{n}{r}~.
\end{equation*}
Thus, there is no integer solution satisfying all the constraints for $m_1=\mu_0+n-r$ and $m_2=\mu_0$. In other words, even when $m_1$ and $m_2$ are both divisible by $r$, and $m_1,m_2\ge\mu_0 \ge n^2/(48r\delta)$, a complete EF allocation of $M$ need not exist.

The proof can be extended to arbitrary $t$ by reproducing the proof for $t=2$ with valuations $\sfV_1$ and $\sfV_2$ such that $\sfv_1(a)=\sfv_2(a)=0$ for all $a> 2$. Intuitively, when nobody values items of type $a$ for $a \ge 3$, only $m_1$ and $m_2$ matter for the existence of complete EF allocations. 
\end{proof}

\subsection{Partial converses of \cref{thm:intro_EF_0modr} and \cref{cor:EF_enoughitems}}\label{sec:result1_converse}
In this subsection, we prove a partial converse of \cref{thm:intro_EF_0modr}. We say an additive valuation $V$, with item-value function $v$, is \emph{irrational} if the ratio $v(a)/v(b)$ is positive and irrational for all distinct $a,b\in[t]$.

\begin{theorem}[Partial converse of \cref{thm:intro_EF_0modr}]\label{thm:partial_converse}
For all pairwise distinct irrational additive valuations $\sfV_1,\ldots,\sfV_d$ and any $M$ with $\b m \not\equiv \b 0^t \pmod r$, there is no complete EF allocation of $M$.
\end{theorem}
\begin{proof}
For any $\sfi\in[d]$, if $i,j\in N_\sfi$, the EF inequalities \eqref{EF-def} between agents $i$ and $j$ become equations because they have identical valuations. That is,
\begin{equation*}
    \sum_{a\in[t]} \sfv_\sfi(a) (x_{i,a}-x_{j,a}) = 0~,
\end{equation*}
where $x_{i,a}$ is the number of type-$a$ items allocated to agent $i$. Since $\sfv_\sfi(a)/\sfv_\sfi(b)$ is positive and irrational for all distinct $a,b\in[t]$, the only integral solution to the above equation is $x_{i,a}=x_{j,a}$ for all $a\in[t]$. In other words, all agents with identical valuations must receive the same set of items.

Let $\sfx_{\sfi,a}$ be the number of type-$a$ items allocated to every agent in $N_\sfi$, i.e., $x_{i,a} = \sfx_{\sfi,a}$ for any $i\in N_\sfi$ and $a\in[t]$. They satisfy the completeness constraint \eqref{comp-def'} given by
\begin{equation*}
    \sum_{\sfi\in[d]} n_\sfi \sfx_{\sfi,a} = m_a~,\quad \forall a\in[t]~.\tag{\ref*{comp-def'}}
\end{equation*}
However, there is no integral solution to the above equation because $m_a \not\equiv 0 \pmod r$ for some $a\in[t]$ (by hypothesis). This means, for any $M$ with $\b m \not\equiv \b 0^t \pmod r$, there is no complete EF allocation of $M$.
\end{proof}

\begin{corollary}[Partial converse of \cref{cor:EF_enoughitems}]\label{cor:partial_converse}
For all pairwise distinct irrational additive valuations, if $r \ne 1$, then for every $\mu\in \mathbb N$, there exists an $M$ with $\b m \ge \mu \b 1^t$ such that there is no complete EF allocation of $M$.
\end{corollary}
\begin{proof}
Since $r\ne1$, choose $M$ such that $\b m \ge \mu \b 1^t$ and $\b m \not\equiv \b 0^t \pmod r$. Then the statement follows from \cref{thm:partial_converse}.
\end{proof}

Note that \cref{thm:partial_converse} is only a partial converse of \cref{thm:intro_EF_0modr}. A full converse of \cref{thm:intro_EF_0modr} would be: \emph{For all pairwise distinct additive valuations, for all $\b k\in (\mathbb Z/r\mathbb Z)^t$, if $\b k\not\equiv \b 0^t \pmod r$, then for all $\mu\in \mathbb N$, there is an $M$ with $\b m \ge \mu \b 1^t$ and $\b m \equiv \b k \pmod r$ such that there is no complete EF allocation of $M$.}
The following example illustrates why this cannot hold.

\begin{example}
For $d=2$, $n_1=n_2=2$, and $t=2$, consider the item-value functions $\sfv_1(a)=1$ and $\sfv_2(a)=a$ for $a\in[2]$. Clearly, the additive valuations $\sfV_1$ and $\sfV_2$ are distinct, and $r=2$. However, we show below that there is a $\mu\in\mathbb N$ such that as long as $m_1,m_2\ge \mu$, there is always a complete EF allocation of $M$.

By \cref{thm:intro_EF_0modr}, there is a $\mu_{0,0}\in\mathbb N$ such that as long as $m_1,m_2\ge \mu_{0,0}$ and $m_1,m_2$ are both even, there is a complete EF allocation of $M$. The problematic cases are when at least one of $m_1,m_2$ is odd. In all the cases below, we assume that $m_1,m_2$ are large enough so that $m_1',m_2'$ are non-negative.
\begin{itemize}
    \item \underline{$m_1$ and $m_2$ are both odd}: Let $m_1' = m_1 - 1$ and $m_2' = m_2 - 3$. Since $m_1'$ and $m_2'$ are both even, by \cref{thm:intro_EF_0modr}, there is a $\mu_{1,1}\in\mathbb N$ such that as long as $m_1',m_2'\ge \mu_{1,1}$, there is a complete EF allocation of $M'$, say $\mathcal X' = \langle (x_{i,1}',x_{i,2}'):i\in[4]\rangle$. Consider the following complete allocation of $M$,
    \begin{align*}
        &x_{1,1} = x_{1,1}'+1~,&& x_{2,1} = x_{2,1}'~,&& x_{3,1} = x_{3,1}'~,&& x_{4,1} = x_{4,1}'~,
        \\
        &x_{1,2} = x_{1,2}'~,&& x_{2,2} = x_{2,2}'+1~,&& x_{3,2} = x_{3,2}'+1~,&& x_{4,2} = x_{4,2}'+1~.
    \end{align*}
    It can be verified that $\mathcal X = \langle (x_{i,1},x_{i,2}):i\in[4]\rangle$ is a complete EF allocation of $M$.
    
    \item \underline{$m_1$ is even and $m_2$ is odd}: Let $m_1' = m_1 - 6$ and $m_2' = m_2 - 1$. Since $m_1'$ and $m_2'$ are both even, by \cref{thm:intro_EF_0modr},  there is a $\mu_{0,1}\in \mathbb N$ such that as long as $m_1',m_2'\ge \mu_{0,1}$, there is a complete EF allocation of $M'$, say $\mathcal X' = \langle (x_{i,1}',x_{i,2}'):i\in[4]\rangle$. Consider the following complete allocation of $M$,
    \begin{align*}
        &x_{1,1} = x_{1,1}'+2~,&& x_{2,1} = x_{2,1}'+2~,&& x_{3,1} = x_{3,1}'+2~,&& x_{4,1} = x_{4,1}'~,
        \\
        &x_{1,2} = x_{1,2}'~,&& x_{2,2} = x_{2,2}'~,&& x_{3,2} = x_{3,2}'~,&& x_{4,2} = x_{4,2}'+1~.
    \end{align*}
    It can be verified that $\mathcal X = \langle (x_{i,1},x_{i,2}):i\in[4]\rangle$ is a complete EF allocation of $M$.
    
    \item \underline{$m_1$ is odd and $m_2$ is even}: Let $m_1' = m_1 - 7$ and $m_2' = m_2 - 4$. Since $m_1'$ and $m_2'$ are both even, by \cref{thm:intro_EF_0modr},  there is a $\mu_{1,0}\in\mathbb N$ such that as long as $m_1',m_2'\ge \mu_{1,0}$, there is a complete EF allocation of $M'$, say $\mathcal X' = \langle (x_{i,1}',x_{i,2}'):i\in[4]\rangle$. Consider the following complete allocation of $M$,
    \begin{align*}
        &x_{1,1} = x_{1,1}'+3~,&& x_{2,1} = x_{2,1}'+2~,&& x_{3,1} = x_{3,1}'+2~,&& x_{4,1} = x_{4,1}'~,
        \\
        &x_{1,2} = x_{1,2}'~,&& x_{2,2} = x_{2,2}'+1~,&& x_{3,2} = x_{3,2}'+1~,&& x_{4,2} = x_{4,2}'+2~.
    \end{align*}
    It can be verified that $\mathcal X = \langle (x_{i,1},x_{i,2}):i\in[4]\rangle$ is a complete EF allocation of $M$.
\end{itemize}
Therefore, there is a $\mu\ge\max\{\mu_{0,0}~,~\mu_{1,1}+3~,~\mu_{0,1}+6~,~\mu_{1,0}+7\}$ such that as long as $m_1,m_2\ge \mu$, there is a complete EF allocation of $M$.
\end{example}

\subsection{Full converse of \cref{cor:EF_enoughitems} for $d=1$}\label{sec:fullConverse}
In this subsection, we prove a full converse of \cref{cor:EF_enoughitems} for the special case of $d=1$.

\begin{theorem}
For every additive valuation $\sfV$, if $n\ge 2$ and every agent has valuation $\sfV$, then for any $\mu\in \mathbb N$, there is an $M$ with $\b m \ge \mu\b 1^t$ such that there is no complete EF allocation of $M$.
\end{theorem}

\begin{proof}
Let $x_{i,a}$ be the number of type-$a$ items allocated to agent $i$. Define a new set of variables $\zeta_{i,a}:=x_{i,a}-x_{i+1,a}$ for $i\in[n-1]$ and $a\in[t]$. Since all agents have identical valuations, the EF inequalities \eqref{EF-def} become equations:
\begin{equation*}
    \sum_{a\in[t]} \sfv(a) \zeta_{i,a} = 0~,\qquad \forall i\in[n-1]~,
\end{equation*}
where $\sfv$ is the item-value function associated with $\sfV$. This is the equation of a $(t-1)$-dimensional hyperplane in $\mathbb R^t$, and the set of all integer points on this hyperplane forms a lattice $\Lambda(\sfv)\subset\mathbb R^t$ isomorphic to $\mathbb Z^{t'}$, where $t' \le t-1$. Let $\b \lambda_1,\ldots,\b \lambda_{t'}$ be a basis of $\Lambda(\sfv)$, where $\b \lambda_{a'}
% = (\lambda_{a',1},\ldots,\lambda_{a',t})
\in\mathbb Z^t$ for each $a'\in[t']$.
Then, for each $i\in[n-1]$, the vector $\b \zeta_i = (\zeta_{i,1},\ldots,\zeta_{i,t})\in \mathbb Z^t$ can be written as an integer linear combination of $\b \lambda_{a'}$'s, i.e.,
\begin{equation*}
    \b \zeta_i = \sum_{a'\in[t']} \alpha_{i,a'} \b \lambda_{a'}~,
\end{equation*}
for some $\alpha_{i,a'}\in\mathbb Z$.

Now, the completeness constraint \eqref{comp-def} implies 
\begin{equation}\label{congruence2}
    \sum_{k=1}^{n-1} k \b \zeta_k = \b m - n \b x_n \equiv \b m \pmod n~.
\end{equation}
Note that the right hand side of \eqref{congruence2} generates $n^t$ congruence classes modulo $n$ because each $m_a$ is an independent integer. 
On the other hand, we have
\begin{equation}\label{congruence3}
    \sum_{k=1}^{n-1} k \b \zeta_k = \sum_{a'\in[t']} \left( \sum_{k=1}^{n-1} k \alpha_{k,a'} \right) \b \lambda_{a'} = \sum_{a'\in[t']} \beta_{a'} \b \lambda_{a'}~,
\end{equation}
where $\beta_{a'}\in\mathbb Z$. Since there are only $t'<t$ integer coefficients $\beta_{a'}$, the right hand side of \eqref{congruence3} generates only at most $n^{t'} < n^t$ congruence classes modulo $n$. So, we can always choose $M$ such that $\b m \ge \mu \b 1^t$, and \eqref{congruence2} is not satisfied, so that there is no complete EF allocation of $M$.
\end{proof}

\section{EFX for additive valuations when $t=2$}\label{sec:t=2}

\algtext*{EndIf}% Remove "end if" text
\begin{algorithm}[p!]
\caption{EFX for additive valuations when $t=2$}\label{alg:tequals2}
\begin{algorithmic}[1]
\State $\mathcal X  = \langle X_1, X_2, \ldots, X_n \rangle \gets \langle \varnothing, \varnothing, \ldots, \varnothing\rangle$ \Comment{Allocation}
\State $\b u \gets \b m:= (m_1, m_2) $ \Comment{Unallocated items}
\State $N_1 \gets \{ i \in [n] : v_i(1) > v_i(2)\}$ \Comment{Agents who prefer type-$1$ items}
\State $N_2 \gets [n] \setminus N_1$\\

\While{$(u_1 \ge |N_1|) \wedge (u_2 \ge |N_2|)$} \Comment{\textbf{Step 1}}
\For {$i \in [n]$} \Comment{Give each agent their most preferred item}
    \If{$i\in N_1$}
        $X_i \gets X_i \uplus \{g_1 \}$; $u_1 \gets u_1 - 1$
    \Else
        \ $X_i \gets X_i \uplus \{g_2 \}$; $u_2 \gets u_2 - 1$
    \EndIf
\EndFor 
\EndWhile
\If{$(u_1 = 0)\wedge(u_2 = 0)$}
    \Return $\mathcal X$
\EndIf \\

\If{$u_1 < |N_1|$}
    $(a,b) \gets (1,2)$ \Comment{\textbf{Step 2}}
\Else
    \ $(a,b) \gets (2,1)$ \Comment{$a$ is the type of items that run out this round}
\EndIf
\State $R \gets$ Arg-sort($\{v_i(a)/v_i(b):i\in[n]\}$)
\State $N_a^+ \gets \{R[1], R[2],\ldots,R[u_a]\}$ \Comment{$N_a^+\subseteq N_a$ are the $u_a$ agents that most prefer $g_a$ to $g_b$}
\For{$i \in N_a^+$}
    \State $X_i \gets X_i \uplus \{g_a \}$; $u_a \gets u_a - 1$
\EndFor 
\If{$u_b = 0$}
    \Return $\mathcal X$
\EndIf \\

\While{$\forall j \in N_a^+, ~\text{Out-degree}(j, G_{\mathcal X})=0$} \Comment{\textbf{Step 3}}
    \For{$i \in [n]\setminus N_a^+$}
        \State $X_i \gets X_i \uplus \{g_b \}$; $u_b \gets u_b - 1$ \Comment{Give a $g_b$ to every agent not in $N_a^+$}
        \If{$u_b = 0$}
            \Return $\mathcal X$
        \EndIf
    \EndFor 
\EndWhile %\Comment{By the time someone in $N_a^+$ envies another, no one envies anyone in $N_a^+$}
\\

\While{true} \Comment{\textbf{Step 4}}
    \For{$i \in N_a^+$}
        \State $X_i \gets X_i \uplus \{g_b \}$; $u_b \gets u_b - 1$ \Comment{Give a $g_b$ to every agent in $N_a^+$}
        \If{$u_b = 0$}
            \Return $\mathcal X$
        \EndIf
    \EndFor
    \For{$i \in [n]\setminus N_a^+$}
        \State $X_i \gets X_i \uplus \{g_b \}$; $u_b \gets u_b - 1$ \Comment{Give a $g_b$ to every agent not in $N_a^+$}
        \If{$u_b = 0$}
            \Return $\mathcal X$
        \EndIf
    \EndFor
\EndWhile
\end{algorithmic}
\end{algorithm}

As mentioned in the introduction, when $t=1$ and valuations are additive, a complete EFX allocation always exists.
In this section, we give an algorithm (see \cref{alg:tequals2}) that returns a complete EFX allocation when $t=2$ and valuations are additive.
In the algorithm, $N_1$ denotes the set of agents who prefer a type-$1$ item to a type-$2$ item, $N_2:=[n]\setminus N_1$ denotes the set of remaining agents, and $\b u:=(u_1,u_2)$ denotes the multiset of unallocated items at any point.
The function ``Arg-sort'' returns the array of indices (i.e., arguments) that sorts the input set in decreasing order.
The function ``Out-degree'' returns the out-degree of the input vertex in the input directed graph.

We now formally prove that \cref{alg:tequals2} returns a complete EFX allocation.

\begin{theorem}\label{thm:t2}
\cref{alg:tequals2} returns a complete EFX allocation.
\end{theorem}
\begin{proof}
The algorithm is divided into four steps.
We use $\mathcal{X}^{(\ell)}$ to denote the allocation at the end of step $\ell\in[4]$.
For any $i\in [n]$, and $\ell\in [4]$, we use $X_i^{(\ell)}$ and $\b x_i^{(\ell)} = (x_{i,1}^{(\ell)},x_{i,2}^{(\ell)})$ interchangeably to denote the multiset of items allocated to agent $i$ at the end of step $\ell$.
We also use $p$, $q$ and $r$ to denote the number of rounds (i.e., number of executions of while loop) in Steps $1$, $3$ and $4$ respectively.

\underline{\textbf{Step 1}}: In this step, we give each agent their most preferred item in a round-robin way.\footnote{For agents who value both type-$1$ and type-$2$ items equally, we assume, without loss of generality, that their most preferred item is of type $2$.}
Each round in this step consists of allocating one item to every agent.
The step ends when, at the end of a round, there are not enough unallocated items to complete another round, i.e., when $u_1 < |N_1|$, or $u_2 < |N_2|$.
Note that the number of rounds in this step is $p = \min\{ \lfloor m_1/|N_1| \rfloor, \lfloor m_2/|N_2| \rfloor \} \ge 0$.
The allocation at the end of this step is
\begin{equation*}
    \b x_i^{(1)} =
    \begin{cases}
        (p,0)~,\quad &\text{for}\quad i\in N_1~,\\
        (0,p)~,\quad &\text{for}\quad i\in N_2~.
    \end{cases}
\end{equation*}
Since all agents have equal number of items, and every item they have is their most preferred, this allocation is EF, and hence, EFX.
If there are no more unallocated items, then this is a complete EFX allocation, and we are done.
If not, we proceed to Step 2.

\underline{\textbf{Step 2}}: Without loss of generality, let us assume that $u_1 < |N_1|$ (this means $(a,b)=(1,2)$ in the algorithm).
Let us order the agents in the decreasing order of the ratios $v_i(1)/v_i(2)$.\footnote{Recall that we assume that every agent has a positive value for at least one type of item.
Hence, there are no ratios of the (indeterminate) form $\frac{0}{0}$.
It is, however, possible for a ratio to be $\infty$.
Any two $\infty$'s are considered equal.}
The set of first $u_1$ agents in this ordering, who most prefer the type-$1$ items, is denoted as $N_1^+$.
Let us also define $N_1^- := N_1 \setminus N_1^+$.
In this step, the agents in $N_1^+$ are all given one type-$1$ item each.
The allocation at the end of this step is
\begin{equation*}
    \b x_i^{(2)} =
    \begin{cases}
        (p+1,0)~,\quad &\text{for}\quad i\in N_1^+~,\\
        (p,0)~,\quad &\text{for}\quad i\in N_1^-~,\\
        (0,p)~,\quad &\text{for}\quad i\in N_2~.
    \end{cases}
\end{equation*}
Let us show that this allocation is EFX.

\begin{enumerate}
    \item There is no envy between any two agents in $N_1^+$ because they have the same sets.
    \item Let $i\in N_1^+$ and $j\in[n]\setminus N_1^+$.
    Then, $i$ does not envy $j$ because $|X_i^{(2)}|=|X_j^{(2)}|+1$, and every item in $X_i^{(2)}$ is $i$'s most preferred (type-$1$).
    On the other hand, $j$ does not envy $i$ after removing any item from $X_i^{(2)}$ because $|X_j^{(2)}|=|X_i^{(2)}|-1$, and every item in $X_j^{(2)}$ is $j$'s most preferred.
    \item There is no envy between any two agents in $[n]\setminus N_1^+$ because there was no envy between them in $\mathcal X^{(1)}$.
\end{enumerate}
If there are no more unallocated items, then this is a complete EFX allocation, and we are done.
If not, we proceed to Step 3.

\underline{\textbf{Step 3}}: At this point, all the unallocated items are of type $2$.
In this step, we give these items to all agents in $[n]\setminus N_1^+$ in a round-robin way.
The step ends when an agent in $N_1^+$ envies an agent in $[n]\setminus N_1^+$ at the end of a round, or when we run out of unallocated items.
We show that the allocation at the end of this step is EFX in both cases.
Note that the number of rounds in this step is $q\ge1$ because no agent in $N_1^+$ envies any agent in $[n]\setminus N_1^+$ at the end of Step 2, and Step 3 is executed only if there are some unallocated items at the end of Step 2.
\begin{itemize}
    \item \underline{\emph{Case (i)}}: An agent in  $N_1^+$ envies an agent in $[n]\setminus N_1^+$ at the end of the $q$th round.
    The allocation in this case is
    \begin{equation}\label{eq:X3}
        \b x_i^{(3)} =
        \begin{cases}
            (p+1,0)~,\quad &\text{for}\quad i\in N_1^+~,\\
            (p,q)~,\quad &\text{for}\quad i\in N_1^-~,\\
            (0,p+q)~,\quad &\text{for}\quad i\in N_2~.
        \end{cases}
    \end{equation}
    \begin{enumerate}
        \item There is no envy between any two agents in $N_1^+$ because they have the same sets.
        
        \item Let $i\in N_1^+$ and $j\in[n]\setminus N_1^+$.
        Note that $X_i^{(3)} = X_i^{(2)}$, and $X_j^{(3)}\supsetneq X_j^{(2)}$ because $q\ge1$.
        Since $\mathcal{X}^{(2)}$ is an EFX allocation, $j$ does not envy $i$ after removing any item from $X_i^{(3)}$.
        Note that before the $q$th round, $i$ did not envy $j$.
        So we have $V_i(X_i^{(3)})\ge V_i(X_j^{(3)}\setminus \{g_2\})$, where $g_2\in X_j^{(3)}$ because $q\ge1$.
        Since $v_i(1)\ge v_i(2)$, $i$ does not envy $j$ after removing any item from $X_j^{(3)}$.
        
        \item There is no envy between any two agents in $[n]\setminus N_1^+$ because they did not envy each other in $\mathcal X^{(2)}$, and they received the same multiset of items in this step.
    \end{enumerate}
    
    \item \underline{\emph{Case (ii)}}: We run out of unallocated items in the $q$th round.
    The allocation in this case is
    \begin{equation*}
        \b x_i^{(3)} =
        \begin{cases}
            (p+1,0)~,\quad &\text{for}\quad i\in N_1^+~,\\
            (p,q)~\text{or}~(p,q-1)~,\quad &\text{for}\quad i\in N_1^-~,\\
            (0,p+q)~\text{or}~(0,p+q-1)~,\quad &\text{for}\quad i\in N_2~.
        \end{cases}
    \end{equation*}
    \begin{enumerate}
        \item There is no envy between any two agents in $N_1^+$ because they have the same sets.
        
        \item Let $i\in N_1^+$ and $j\in[n]\setminus N_1^+$.
        Note that $X_i^{(3)} = X_i^{(2)}$, and $X_j^{(3)}\supseteq X_j^{(2)}$ because $q\ge1$.
        Since $\mathcal{X}^{(2)}$ is an EFX allocation, $j$ does not envy $i$ after removing any item from $X_i^{(3)}$.
        Note that before the $q$th round, $i$ did not envy $j$.
        So, if $i$ envies $j$ in $\mathcal X^{(3)}$, $j$ must have received a type-$2$ item in the $q$th round.
        Combining these, we have $V_i(X_i^{(3)})\ge V_i(X_j^{(3)}\setminus \{g_2\})$.
        Since $v_i(1)\ge v_i(2)$, $i$ does not envy $j$ after removing any item from $X_j^{(3)}$.
        
        \item Let $i,j\in [n]\setminus N_1^+$.
        If $|X_i^{(3)}| = |X_j^{(3)}|$, then they do not envy each other because they did not envy each other in $\mathcal X^{(2)}$, and they received the same multiset of items in this step.
        If $|X_i^{(3)}| = |X_j^{(3)}| + 1$, then $i$ must have received a type-$2$ item in the $q$th round, while $j$ did not.
        Then, $i$ does not envy $j$ because $i$ did not envy $j$ in $\mathcal X^{(2)}$, and $i$ received more type-$2$ items than $j$ in this step.
        Whereas, $j$ does not envy $i$ after removing any item from $X_i^{(3)}$ because $|X_j^{(3)}| = |X_i^{(3)}| - 1$, and $X_j^{(3)}$ contains at least as many of $j$'s most preferred items as $X_i^{(3)}\setminus\{g\}$ for any $g\in X_i^{(3)}$.
    \end{enumerate}
\end{itemize}
If there are no more unallocated items, then this is a complete EFX allocation in both cases, and we are done.
If not, which can happen only when Step 3 ends in Case (i), we proceed to Step 4.

Before proceeding to Step 4, let us state a useful lemma.
\begin{lemma}\label{lem:noenvy}
In the allocation $\mathcal X^{(3)}$ given by \eqref{eq:X3}, i.e., when Step 3 ends in Case (i), no agent in $[n]\setminus N_1^+ = N_1^- \cup N_2$ envies any agent in $N_1^+$.
\end{lemma}
We prove this lemma later, and proceed to Step 4 now.

\underline{\textbf{Step 4}}: In this step, we give type-$2$ items to all agents in a round-robin way.
Each round starts with giving an item to every agent in $N_1^+$, and only then to the remaining agents.
The step ends when we run out of unallocated items.
When this happens, either some agent in $[n]\setminus N_1^+$ received an item in the last round, or not.
We show that the allocation at the end of this step is EFX in both cases.
Note that number of rounds in this step is $r\ge1$ because Step 4 is executed only if there are some unallocated items at the end of Step 3, which, as we mentioned above, can happen only when Step 3 ends in Case (i).
\begin{itemize}
    \item \underline{\emph{Case (i)}}: Some agent in $[n]\setminus N_1^+$ received an item in the $r$th round.
    This means, every agent in $N_1^+$ must have received an item in the $r$th round.
    The allocation in this case is
    \begin{equation*}
        \b x_i^{(4)} =
        \begin{cases}
            (p+1,r)~,\quad &\text{for}\quad i\in N_1^+~,\\
            (p,q+r)~\text{or}~(p,q+r-1)~,\quad &\text{for}\quad i\in N_1^-~,\\
            (0,p+q+r)~\text{or}~(0,p+q+r-1)~,\quad &\text{for}\quad i\in N_2~.
        \end{cases}
    \end{equation*}
    \begin{enumerate}
        \item There is no envy between any two agents in $N_1^+$ because they have the same sets.
        
        \item Let $i\in N_1^+$ and $j\in[n]\setminus N_1^+$.
        We know that $V_i(X_i^{(3)})\ge V_i(X_j^{(3)}\setminus \{g_2\})$ in \eqref{eq:X3}, i.e,
        \[ V_i(p+1,0)\ge V_i(p,q-1) \ge V_i(0,p+q-1)~, \]
        where the second inequality follows from $v_i(1)\ge v_i(2)$.
        This implies that
        \[V_i(p+1,r)\ge V_i(p,q+r-1) \ge V_i(0,p+q+r-1)~.\]
        Since $v_i(1)\ge v_i(2)$, it follows that $i$ does not envy $j$ after removing any item from $X_j^{(4)}$.
        Now, say $j\in N_2$.
        Since $|X_j^{(4)}|\ge |X_i^{(4)}|-1$, and $X_j^{(4)}$ has only $j$'s most preferred items, $j$ does not envy $i$ after removing any item from $X_i^{(4)}$.
        Next, say $j\in N_1^-$.
        By \cref{lem:noenvy}, we have $V_j(X_j^{(3)})\ge V_j(X_i^{(3)})$ in \eqref{eq:X3}, i.e., $V_j(p,q)\ge V_j(p+1,0)$.
        This implies that $V_j(p,q+r-1)\ge V_j(p+1,r-1)$.
        Since $v_j(1)\ge v_j(2)$, it follows that $j$ does not envy $i$ after removing any item from $X_i^{(4)}$.
        
        \item Let $i,j\in [n]\setminus N_1^+$.
        If $|X_i^{(4)}| = |X_j^{(4)}|$, then they do not envy each other because they did not envy each other in $\mathcal X^{(3)}$ of \eqref{eq:X3}, and they received the same multiset of items in this step.
        If $|X_i^{(4)}| = |X_j^{(4)}| + 1$, then $i$ must have received a type-$2$ item in the $r$th round, while $j$ did not.
        Then, $i$ does not envy $j$ because $i$ did not envy $j$ in $\mathcal X^{(3)}$ of \eqref{eq:X3}, and $i$ received more type-$2$ items than $j$ in this step.
        Whereas, $j$ does not envy $i$ after removing any item from $X_i^{(4)}$ because $|X_j^{(4)}| = |X_i^{(4)}| - 1$, and $X_j^{(4)}$ contains at least as many of $j$'s most preferred items as $X_i^{(4)}\setminus\{g\}$ for any $g\in X_i^{(4)}$.
    \end{enumerate}
    
    \item \underline{\emph{Case (ii)}}: No agent in $[n]\setminus N_1^+$ received an item in the $r$th round.
    The allocation in this case is
    \begin{equation*}
        \b x_i^{(4)} =
        \begin{cases}
            (p+1,r)~\text{or}~(p+1,r-1)~,\quad &\text{for}\quad i\in N_1^+~,\\
            (p,q+r-1)~,\quad &\text{for}\quad i\in N_1^-~,\\
            (0,p+q+r-1)~,\quad &\text{for}\quad i\in N_2~.
        \end{cases}
    \end{equation*}
    \begin{enumerate}
        \item Let $i,j\in N_1^+$.
        If $X_i^{(4)} = X_j^{(4)}$, then they do not envy each other.
        If $X_i^{(4)} = X_j^{(4)} \uplus \{g_2\}$, then $i$ does not envy $j$, while $j$ does not envy $i$ after removing any item from $X_i^{(4)}$ because $v_j(1)> v_j(2)$.
        
        \item Let $i\in N_1^+$ and $j\in[n]\setminus N_1^+$.
        We know that $V_i(X_i^{(3)})\ge V_i(X_j^{(3)}\setminus \{g_2\})$ in \eqref{eq:X3}, i.e,
        \[ V_i(p+1,0)\ge V_i(p,q-1) \ge V_i(0,p+q-1)~, \]
        where the second inequality follows from $v_i(1)\ge v_i(2)$.
        This implies that
        \[V_i(p+1,r-1)\ge V_i(p,q+r-2) \ge V_i(0,p+q+r-2)~.\]
        Since $v_i(1)\ge v_i(2)$, it follows that $i$ does not envy $j$ after removing any item from $X_j^{(4)}$.
        Now, say $j\in N_2$.
        Since $|X_j^{(4)}|\ge |X_i^{(4)}|-1$, and $X_j^{(4)}$ has only $j$'s most preferred items, $j$ does not envy $i$ after removing any item from $X_i^{(4)}$.
        Next, say $j\in N_1^-$.
        By \cref{lem:noenvy}, we have $V_j(X_j^{(3)})\ge V_j(X_i^{(3)})$ in \eqref{eq:X3}, i.e., $V_j(p,q)\ge V_j(p+1,0)$.
        This implies that $V_j(p,q+r-1)\ge V_j(p+1,r-1)$.
        Since $v_j(1)\ge v_j(2)$, it follows that $j$ does not envy $i$ after removing any item from $X_i^{(4)}$.
        
        \item There is no envy between any two agents in $[n]\setminus N_1^+$ because they did not envy each other in $\mathcal X^{(3)}$ of \eqref{eq:X3}, and they received the same multiset of items in this step.
    \end{enumerate}
\end{itemize}
Since there are no more unallocated items, this is a complete EFX allocation in both cases, and we are done.
\end{proof}

Let us now prove \cref{lem:noenvy} to complete the proof of \cref{thm:t2}.
\begin{proof}[Proof of \cref{lem:noenvy}]
Let us first consider any $j\in N_2$, and any $i\in N_1^+$.
In \eqref{eq:X3}, $|X_j^{(3)}|\ge|X_i^{(3)}|$, and every item in $X_j^{(3)}$ is $j$'s most preferred (type-$2$).
Hence, $j$ does not envy $i$, i.e., no agent in $N_2$ envies any agent in $N_1^+$.

Next, we consider any $j\in N_1^-$.
Since we assume that Step 3 ends in Case (i), there is an $i\in N_1^+$ who envies some $k\in [n]\setminus N_1^+$.
Moreover, in \eqref{eq:X3}, $|X_j^{(3)}|=|X_k^{(3)}|$, and $x_{j,1}^{(3)}\ge x_{k,1}^{(3)}$.
It follows that $V_i(X_j^{(3)})\ge V_i(X_k^{(3)})$ because $v_i(1)\ge v_i(2)$.
Hence, $V_i(X_i^{(3)})< V_i(X_j^{(3)})$, i.e., $i$ envies $j$.

Observe that $x_{i,1}^{(3)} > x_{j,1}^{(3)}$ in \eqref{eq:X3}.
Also, by the definition of $N_1^+$, we have $v_i(1)/v_i(2) \ge v_j(1)/v_j(2)$.
(Note that $v_i(1)\ge0$ and $v_j(1)>0$ because $i,j\in N_1$.
Since $x_{i,1}^{(3)} > x_{j,1}^{(3)}$, and $i$ envies $j$, we conclude that $v_i(2)>0$.
It then follows that $v_j(2)\ge v_j(1) \frac{v_i(2)}{v_i(1)} >0$.)
Putting these facts together, for $\mathcal X^{(3)}$ in \eqref{eq:X3}, we have
\begin{align*}
    V_j(X_i^{(3)}) - V_j(X_j^{(3)})
    &= v_j(1) \left(x_{i,1}^{(3)} - x_{j,1}^{(3)}\right) + v_j(2) \left(0 - x_{j,2}^{(3)}\right)
    \\
    &\le \frac{v_j(2)}{v_i(2)}\left[ v_i(1) \left(x_{i,1}^{(3)} - x_{j,1}^{(3)}\right) + v_i(2) \left(0 - x_{j,2}^{(3)}\right) \right]
    \\
    &= \frac{v_j(2)}{v_i(2)} \left[V_i(X_i^{(3)}) - V_i(X_j^{(3)})\right] < 0~.
\end{align*} 
Thus, $j$ does not envy $i$.
Since $X_i^{(3)}=X_{i'}^{(3)}$ for every $i,i'\in N_1^+$ in \eqref{eq:X3}, we conclude that no agent in $N_1^-$ envies any agent in $N_1^+$.
\end{proof}

\cref{alg:tequals2} runs in time polynomial in $n$ and $m$. It takes constant time to assign each item.
The ``Arg-sort'' method takes $\Theta(n \log{n})$ time.
Updating the envy graph after assigning an item takes $O(n)$ time.
The ``Out-degree'' method takes $O(n)$ time.

\section{Alternative proof of EFX for additive valuations when $t=2$}\label{sec:geometric}

As in \cref{sec:t=2}, consider $n$ agents with additive valuations, and $t$ types of items with $m_a$ items of type $a\in[t]$.
In this section, we will provide a geometrical proof of existence of complete EFX allocations for $t=2$ (\cref{thm:intro_t2}), and explain why this proof does not extend immediately to $t\ge 3$.

\subsection{Preliminaries}

Recall that any multiset $X$ with $x_a$ items of type $a\in[t]$ is associated with a nonnegative integer point $\b x := (x_1,\ldots,x_t) = \sum_{a\in[t]}x_a \b e_a\in \mathbb Z_{\ge0}^t$, where $\b e_1,\ldots,\b e_t$ are the standard unit vectors of $\mathbb R^t$.
For any agent $i$, we will use $V_i(X)$ and $V_i(\b x)$ interchangeably.
Similarly, we will denote an allocation as $\mathcal X = \langle X_1,\ldots,X_n\rangle$ or $\mathcal X = \langle \b x_1,\ldots,\b x_n \rangle$.

We will need the concept of \emph{Pareto dominance} of allocations.
We say an allocation $\mathcal Y$ Pareto dominates another allocation $\mathcal X$ if $V_i(Y_i) \ge V_i(X_i)$ for all $i\in[n]$.
If in addition $V_j(Y_j) > V_j(X_j)$ for some $j\in[n]$, then we say $\mathcal Y$ strictly Pareto dominates $\mathcal X$.
Note that Pareto dominance gives a partial order on the set of all allocations.

We will also need a few geometric ideas described below.
Let us extend the additive valuation $V_i$ from $\mathbb Z_{\ge0}^t$ to $\mathbb R^t$ linearly, i.e., define $V_i(\b x):=\sum_{a\in[t]} v_i(a) x_a$ for any $\b x = (x_1,\ldots,x_t)\in \mathbb R^t$.
Consider the hyperplane $L_i:=\{\b x\in \mathbb R^t:V_i(\b x)=V_i(\b x_i)\}$ in $\mathbb R^t$ passing through the nonnegative integer point $\b x_i$, and with nonnegative intercepts on all axes.\footnote{Note that the hyperplanes $L_i$ defined here are different from the hyperplanes $H_\sfi$ defined in \cref{sec:EF}. While the hyperplanes here pass through $\b x_i$ but not through the origin $\b 0^t$, the hyperplanes there pass through the origin.}
So, for each agent $i$, we have the following association: $\b x_i$ corresponds to $X_i$, and $L_i$ corresponds to $v_i$.
Observe that the $x_a$-intercept of $L_i$, denoted as $L_{i,a}$, is inversely proportional to $v_i(a)$.
More precisely, $L_{i,a} = V_i(\b x_i)/v_i(a)$.

We refer to the open half-space $\{\b x\in \mathbb R^t: V_i(\b x)>V_i(\b x_i)\}$ as ``above the hyperplane $L_i$,'' and the closed half-space $\{\b x \in \mathbb R^t: V_i(\b x)\le V_i(\b x_i)\}$ as ``below the hyperplane $L_i$.''
For example, if agent $i$ does (not) envy agent $j$, we say that ``$\b x_j$ is above (below) $L_i$.''
In particular, agent $i$ is not envied by any other agent, i.e., $i$ is a source in the envy graph $G_{\mathcal X}$ if and only if the integer point $\b x_i$ is below $L_j$ for all $j\in[n]$.

We define the following partial order on $\mathbb R^t$: $\b x\le \b y$ if $x_a \le y_a$ for all $a\in [t]$.
The inequality is strict if in addition $x_a < y_a$ for some $a\in[t]$.
For example, if $\b x$ and $\b y$ are integer points associated with sets $X$ and $Y$, then $\b x < \b y$ if and only if $X \subsetneq Y$.
It is useful to note that if $\b x \le \b y$, and $\b y$ is below $L_i$, then $\b x$ is also below $L_i$.
On the other hand, if $\b x > \b y$, and $\b y$ is above $L_i$, then $\b x$ is above $L_i$.

\subsection{Geometrical proof for $t=2$}
We are now ready to prove that complete EFX allocations exist when $t=2$.
In this case, the ambient space is a plane $\mathbb R^2$, and hyperplanes are lines.
So, for every agent $i\in[n]$, a nonnegative integer point $\b x_i=(x_{i,1},x_{i,2})$ corresponds to $X_i$, and a line $L_i$ passing through $\b x_i$, and making nonnegative intercepts on both the axes, corresponds to $v_i$.
An example is shown in \cref{fig:example}.

\begin{figure}
    \centering
    \includegraphics[scale=0.25]{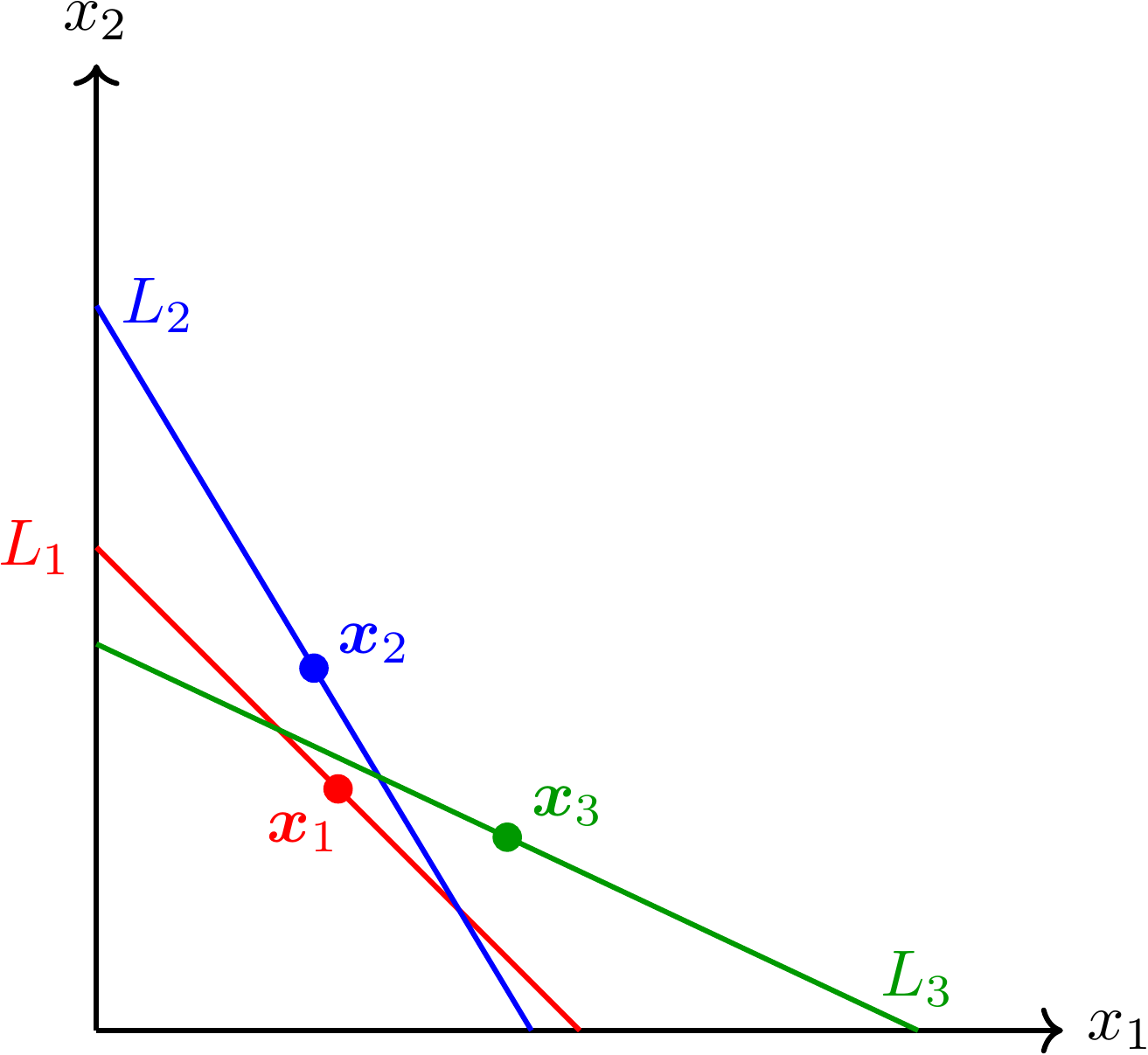}
    \caption{Illustration of points and hyperplanes of agents when $n=3$ and $t=2$.
    The lines (hyperplanes in two dimensions) $L_1,L_2,L_3$ correspond to the valuations, and the points $\b x_1,\b x_2,\b x_3$ correspond to the sets of the three agents.
    Agent $1$ (red) is a source, i.e., no agent envies agent $1$ because $\b x_1$ (red point) is below the rest of the lines $L_2,L_3$ (blue and green lines).}
    \label{fig:example}
\end{figure}

Say $\mathcal X$ is an EFX allocation,\footnote{There is always an EFX allocation because the empty allocation is one.} and $g$ is an unallocated item.
Following \cite{efx_charity}, our strategy is to exhibit a new EFX allocation $\mathcal X'$ that strictly Pareto dominates $\mathcal X$.
Then, when there are no more unallocated items, we are left with a complete EFX allocation.

If giving $g$ to some agent $i$ gives an EFX allocation $\mathcal X'$, with $X_i'=X_i \uplus \{g\}$ and $X_j'=X_j$ for $j\ne i$, then we are done because $\mathcal X'$ strictly Pareto dominates $\mathcal X$.
In what follows, we will assume this is not the case for any $i\in[n]$.

It is well-known that directed cycles in the envy graph of an allocation $\mathcal X$ can be removed while preserving EFX.
This is done by picking a directed cycle and assigning $X_j$ to agent $i$ for every edge $i \to j$ in the directed cycle until there are no more directed cycles.
We refer to this procedure as \emph{de-cycling}.
(See Lemma 2 of \cite{efx_charity} for a formal proof of this.)
In fact, de-cycling yields a strictly Pareto dominant EFX allocation \cite{efx_charity}.
So, without loss of generality, we can assume that the envy graph $G_{\mathcal X}$ is a directed acyclic graph (DAG).
In particular, there is always a source in $G_{\mathcal X}$.

The main result of this section is the following:

\begin{lemma}\label{lem:reachablemea}
Say $g$ is of type $a\in[2]$.
Pick an agent $s\in[n]$ who is a source in $G_{\mathcal X}$, and whose line $L_s$ makes the least intercept along the $x_a$-axis among all the sources, i.e., $L_{s,a} \le L_{s',a}$ for any source $s'$ in $G_{\mathcal X}$.\footnote{If there is a tie, break the tie arbitrarily.}
Then, there is an agent $r\in[n]$ who is 
\begin{enumerate}
    \item a \emph{most envious agent} (MEA) of the multiset $X_s \uplus \{g\}$ \cite{efx_charity}: there is a strict subset $Z_s \subsetneq X_s \uplus \{g\}$ such that agent $r$ envies $Z_s$, but no strict subset of $Z_s$ is envied by any agent.
    
    \item \emph{reachable} from $s$ in $G_{\mathcal X}$ \cite{efx_charity}: there is a directed path $s = i_K \rightarrow i_{K-1} \rightarrow \cdots \rightarrow i_2 \rightarrow i_1 = r$ in $G_{\mathcal X}$,\footnote{We allow $K=1$, i.e., $s$ is trivially reachable from itself.
    However, this does not mean there are self loops in $G_{\mathcal X}$.}
\end{enumerate}
We say $r$ is a \emph{reachable MEA} for $s$ with respect to $g$.
\end{lemma}
The two notions in the lemma, \emph{most envious agent} and \emph{reachability}, were defined in \cite{efx_charity}.

It follows from \cref{lem:reachablemea} that the allocation $\mathcal X'$, given by $X_{i_{p+1}}'=X_{i_p}$ for $p\in[K-1]$, $X_r'=Z_s$, and $X_i'=X_i$ for the remaining $i$, is an EFX allocation.
Moreover, $\mathcal X'$ strictly Pareto dominates $\mathcal X$ because $V_{i_p}(X_{i_p}')> V_{i_p}(X_{i_p})$ for $p\in[K]$, and $V_i(X_i')=V_i(X_i)$ for the remaining $i$.

Therefore, EFX exists when $t=2$ assuming  \cref{lem:reachablemea}.
Let us prove the lemma now.

\begin{proof}[Proof of \cref{lem:reachablemea}]

Our proof is divided into two steps: in Step 1, we will find a suitable MEA of $X_s \uplus \{g\}$, and then in Step 2, we show that this MEA is reachable from $s$ in $G_{\mathcal X}$.
For concreteness, let us assume that $g$ is of type $2$.

\textbf{\underline{Step 1}}: Consider the nonnegative integer point $\b x_s + \b e_2 = (x_{s,1},x_{s,2}+1)$ associated with the multiset $X_s \uplus \{g\}$.
Any strict subset $Z_s \subsetneq X_s\uplus \{g\}$ is associated with a nonnegative integer point $\b z_s < \b x_s + \b e_2$.
First, we make a few remarks about what $Z_s$ can be.
Removing a type-$2$ item from $X_s \uplus \{g\}$ is equivalent to removing $g$.
But then, $Z_s \subseteq X_s$, and since $s$ is a source, no agent envies $Z_s$.
That is, there cannot be a MEA with $Z_s \subseteq X_s$.
So, we can remove only type-$1$ items from $X_s \uplus \{g\}$.
In other words, $z_{s,2} = x_{s,2} +1$, and $z_{s,1} < x_{s,1}$.
Note that $x_{s,1}>0$ because, if not, then giving $g$ to agent $s$ would give a strictly Pareto dominant EFX allocation because $s$ is a source.

We will now find a suitable MEA of $X_s \uplus \{g\}$.
Let $N_s := \{i\in[n]: L_{i,2} < x_{s,2} + 1\}$, and let $\ell_s := \{\b x \in\mathbb R^2: x_2 = x_{s,2} + 1\}$ be the line parallel to $x_1$-axis that meets $x_2$-axis at $\b q_s = (0,x_{s,2}+1)$.
Recall that $L_{i,2}$ is the $x_2$-intercept of $L_i$.
There are two scenarios, which are illustrated in \cref{fig:mea}:
\begin{enumerate}
    \item \underline{$N_s$ is not empty}: Pick an agent $r\in N_s$ such that $L_{r,2} \le L_{i,2}$ for every $i\in [n]$.\footnote{If there is a tie, break the tie arbitrarily unless when $s\in N_s$ and $s$ is part of the tie, in which case, pick $r=s$.}
    In other words, pick an agent $r\in N_s$ whose line $L_r$ makes the least intercept on the $x_2$-axis.
    Let $\b z_s = \b q_s$.
    Since $r\in N_s$, $\b z_s$ is above $L_r$, and so $r$ envies $Z_s$.
    Moreover, for any $Y_s \subsetneq Z_s$, we have $\b y_s < \b z_s$, which means $\b y_s < \b x_s$ because $y_{s,1} = z_{s,1} = 0 < x_{s,1}$ and $y_{s,2} < z_{s,2} = x_{s,2} + 1$.
    So, $Y_s$ is a strict subset of $X_s$, and hence, it is not envied by any agent because $\mathcal X$ is an EFX allocation.
    Therefore, $r$ is a MEA of $X_s \uplus \{g\}$.
    
    \item \underline{$N_s$ is empty}: Let $\b p_i$ be the point of intersection of $L_i$ and $\ell_s$ in the positive quadrant for each $i\in[n]$.
    Pick an agent $r\in[n]$ such that $p_{r,1} \le p_{i,1}$ for all $i\in [n]$.\footnote{If there is a tie, break the tie arbitrarily unless when $s$ is part of the tie, in which case, pick $r=s$.}
    In other words, pick an agent $r\in[n]$ whose $\b p_r$ has the least $x_1$-coordinate.
    Then, $\b p_r$ is below $L_i$ for all $i\in [n]$.
    Moreover, $p_{r,1} < x_{s,1}-1$ because, if not, then giving $g$ to agent $s$ would give a strictly Pareto dominant EFX allocation because $s$ is a source.
    So, there is an integer point $\b z_s$ such that $\b p_r < \b z_s < \b x_s + \b e_2$, and $\b z_s - \b e_1 \le \b p_r$.
    Then, $r$ envies $Z_s$ because $\b p_r < \b z_s$ implies that $\b z_s$ is above $L_r$.
    Consider any $\b y_s < \b z_s$ associated with a strict subset $Y_s \subsetneq Z_s$. There are two options for $Y_s$:
    \begin{itemize}
        \item If $y_{s,2} \le x_{s,2}$ (i.e., some type-$2$ item is removed from $Z_s$), then $\b y_s < \b x_s$ because $y_{s,1} \le z_{s,1} < x_{s,1}$.
        Since $\mathcal X$ is EFX, no agent envies $Y_s \subsetneq X_s$.
        
        \item If $y_{s,2} = x_{s,2}+1 = z_{s,2}$ (i.e., no type-$2$ item is removed from $Z_s$), then $\b y_s \le \b p_r$ because $\b y_s < \b z_s \implies y_{s,1} \le z_{s,1}-1$.
        Since $\b p_r$ is below $L_i$ for all $i\in [n]$, so is $\b y_s$, and hence no agent envies $Y_s$.
    \end{itemize}
    Therefore, $r$ is a MEA of $X_s \uplus \{g\}$.
    It is useful to note that if $r\ne s$, then $L_{r,2} < L_{s,2}$ because $\b p_s$ and $\b x_s$ are above and below $L_r$ respectively (see the right side of \cref{fig:mea} for a geometrical proof of this).
\end{enumerate}

\begin{figure}
    \centering
    \hfill\includegraphics[scale=0.25]{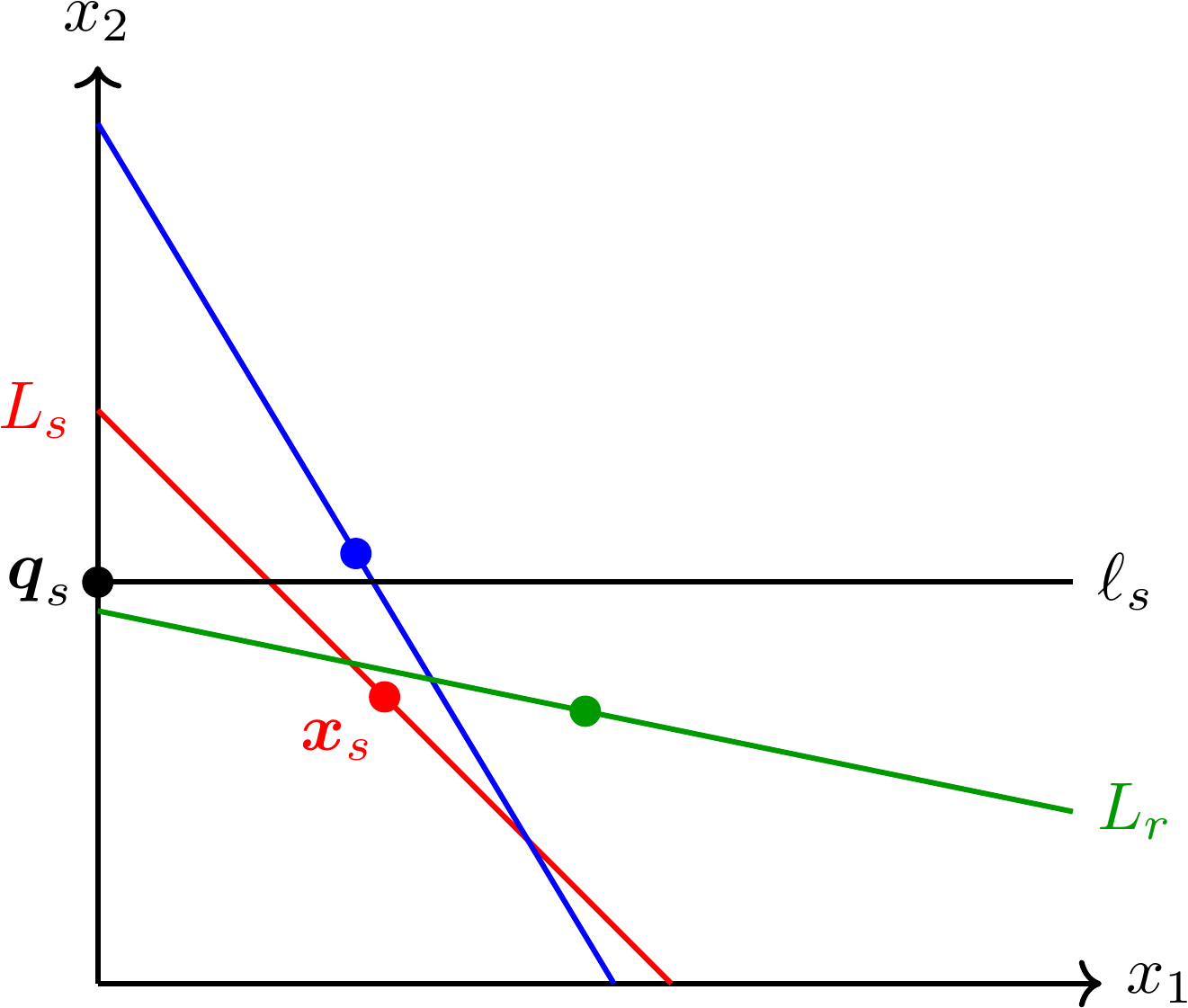}\hfill
    \includegraphics[scale=0.25]{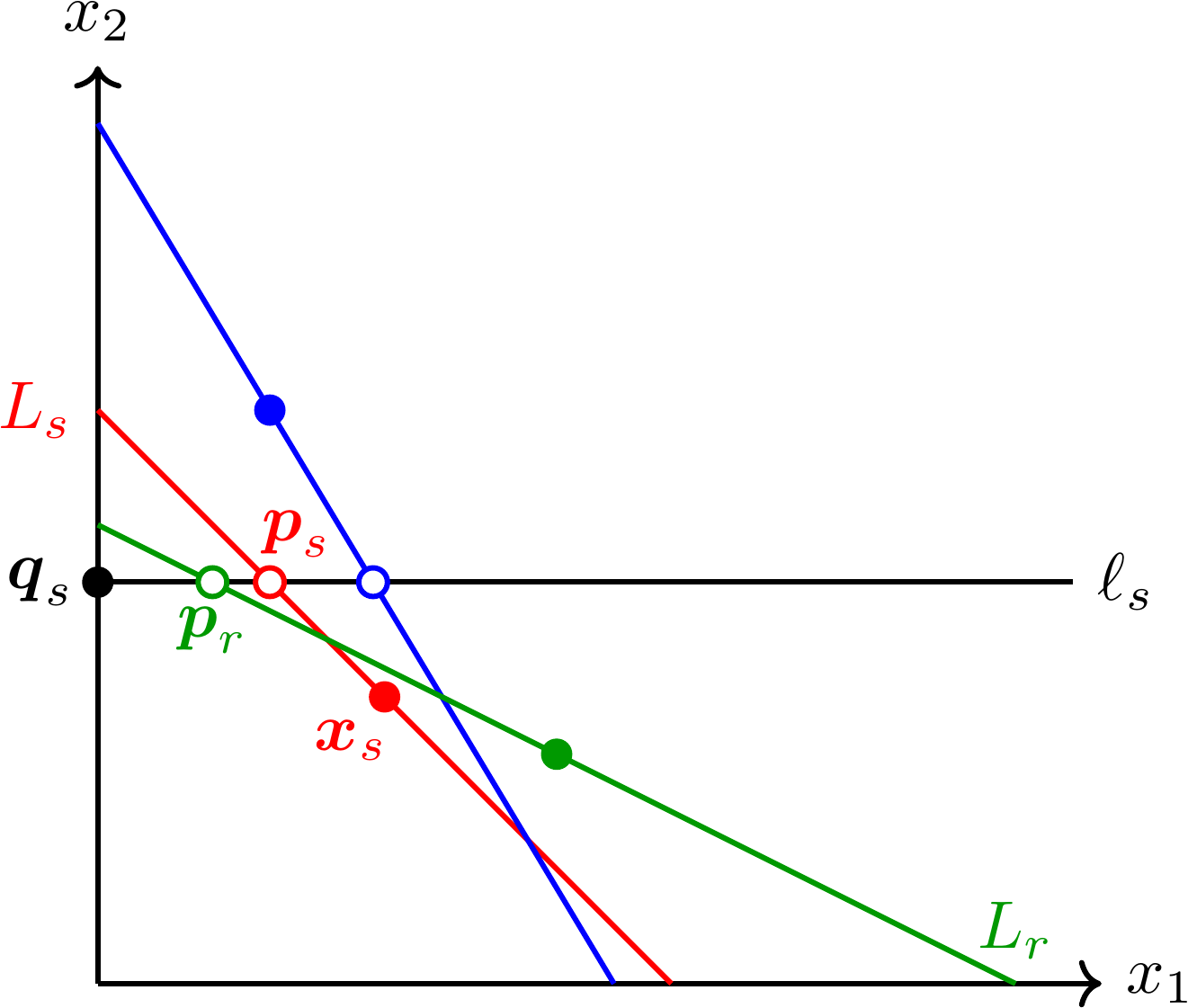}\hfill\hfill\hfill
    \\
    \hfill$N_s$ is not empty\hfill
    $N_s$ is empty\hfill\hfill
    \caption{Two scenarios when finding a suitable MEA $r$ of $X_s \uplus\{g\}$ assuming $g$ is of type $2$.
    Here $n=3$.
    Agent $s$ with point $\b x_s$ (solid red dot) and line $L_s$ (red line) is a source because $\b x_s$ is below all the other lines (green and blue lines).
    The black line is $\ell_s$, which meets the $x_2$-axis at $\b q_s = (0,x_{s,2}+1)$.
    On the left, $N_s$ is not empty because $L_r$ (green line) has a smaller $x_2$-intercept than $x_{s,2}+1$.
    Since the green line has the smallest $x_2$-intercept, we pick the green agent to be agent $r$.
    On the right, $N_s$ is empty because all agents have larger $x_2$-intercept than $x_{s,2}+1$.
    Since the hollow green dot (point of intersection of the green line and $\ell_s$) has the smallest $x_1$-coordinate, we pick the green agent to be agent $r$.
    It is clear from the figure that $L_{r,2} < L_{s,2}$, i.e., the $x_2$-intercept of $L_r$ is less than that of $L_s$ because, when $r\ne s$, $\b p_s$ (hollow red dot) and $\b x_s$ are above and below $L_r$ respectively.}
    \label{fig:mea}
\end{figure}

\textbf{\underline{Step 2}}: We will now prove that $i_1=r$ found in Step 1 is reachable from $s$ in $G_{\mathcal X}$ by induction.
If $i_1=s$, we are done.
If $i_1\ne s$, then recall that $L_{i_1,2} < L_{s,2}$ in both scenarios of Step 1.
This concludes the base case.
Let us move on to the induction step.

Assume that $i_1,i_2,\ldots, i_{K-1}$ are distinct agents in $[n]\setminus\{s\}$ such that $i_{K-1} \rightarrow \cdots\rightarrow i_2 \rightarrow i_1$ is a directed path, and $L_{i_p,2}<L_{s,2}$ for all $p\in[K-1]$.
Note that if there is an edge $s\rightarrow i_p$ for some $p\in[K-1]$, then we are done, so we assume that this is not the case.
Then, we will show that there is an $i_K\in [n]\setminus\{s,i_1,\ldots,i_{K-1}\}$ such that $i_K \rightarrow \cdots\rightarrow i_2 \rightarrow i_1$ is a directed path, and $L_{i_p,2}<L_{s,2}$ for all $p\in[K]$.

Since $L_{i_{K-1},2}<L_{s,2}$, $i_{K-1}$ is not a source because $s$ has the least $x_2$-intercept among all sources by hypothesis.
So, there is an edge $i_K \rightarrow i_{K-1}$ for some agent $i_K$.
By induction hypothesis, there is no edge $s \rightarrow i_{K-1}$.
Also, for all $p\in[K-2]$, there is no edge $i_p \rightarrow i_{K-1}$ because there are no directed cycles in $G_{\mathcal X}$.
Therefore, $i_K \in [n]\setminus\{s,i_1,\ldots,i_{K-1}\}$, and $i_K \rightarrow \cdots\rightarrow i_2 \rightarrow i_1$ is a directed path.

Note that $\b x_s$ is below $L_{i_{K-1}}$ and $L_{i_K}$ because $s$ is a source.
Note also that that $\b x_{i_{K-1}}$ is above $L_{i_K}$ because there is an edge $i_K\rightarrow i_{K-1}$, but it is below $L_s$ because there is no edge $s\rightarrow i_{K-1}$.
From these, we conclude that $L_{i_K,2} < L_{s,2}$ (see \cref{fig:reachable} for a geometrical proof of this).
This concludes the induction step.

\begin{figure}
    \centering
    \includegraphics[scale=0.25]{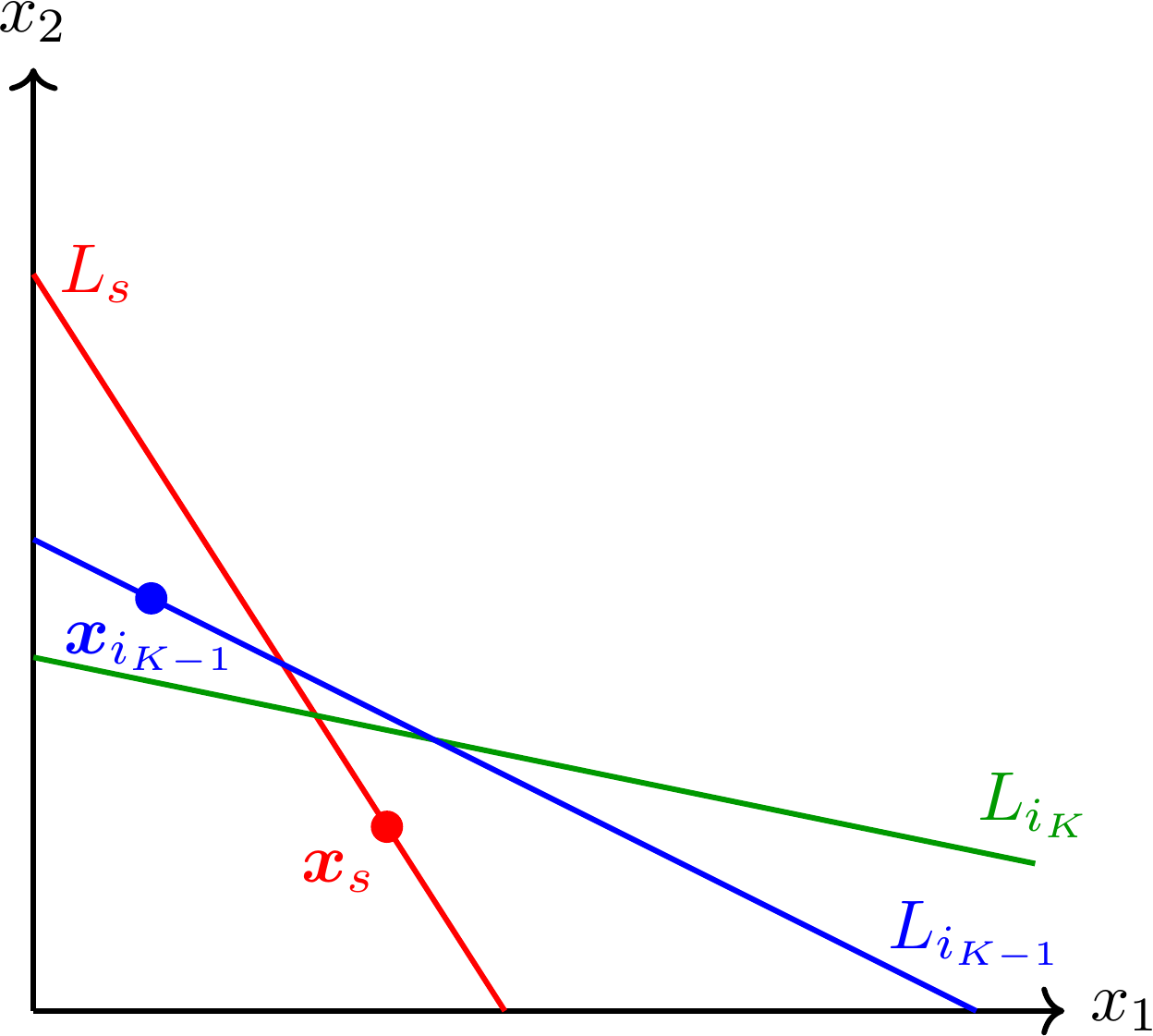}
    \caption{As explained in the main text, $\b x_s$ (red dot) is below $L_{i_{K-1}}$ (blue line) and $L_{i_K}$ (green line), whereas $\b x_{i_{K-1}}$ (blue dot) is above $L_{i_K}$, but below $L_s$ (red line).
    A candidate for $L_{i_K}$ satisfying these restrictions is shown above.
    It follows that $L_{i_K,2} < L_{s,2}$, i.e., the $x_2$-intercept of $L_{i_K}$ is less than that of $L_s$.}
    \label{fig:reachable}
\end{figure}

Since there are only $n$ agents, the above process has to stop, i.e., there is a directed path from $s$ to $r$ in $G_{\mathcal X}$.
Therefore, $r$ is reachable from $s$ in $G_{\mathcal X}$, and the lemma is proved.
\end{proof}

\subsection{Counterexample to \cref{lem:reachablemea} for $t\ge 3$}\label{sec:counterex}
%\gnote{This title seems to suggest that we have counterexample to existence of EFX when $t\ge3$. We should come up with a better title.}

When $t\ge 3$, \cref{lem:reachablemea} is no longer true.
In fact, there need not be a reachable MEA for \emph{any} source.
We prove this by exhibiting a counterexample in the case of $t=3$.
Since $t=3$ is a special case of $t\ge 3$, the same counterexample works for all $t\ge 3$.

Consider $n=3$ agents, $A$, $B$, and $C$, with additive valuations given by the functions $v_A$, $v_B$, and $v_C$ shown in \cref{tbl:noreachablemea}.
Say the numbers of items of each type to be allocated are $m_1 = 9$, $m_2=9$, and $m_3=11$.
It is easy to check that the allocation $\mathcal X = \langle \b x_A,\b x_B,\b x_C\rangle$, where $\b x_A = (3,5,0)$, $\b x_B = (6,3,0)$, and $\b x_C = (0,0,11)$, is EFX.
The envy graph $G_{\mathcal X}$ is shown in \cref{fig:envygraph}.
We see that $A$ and $B$ are sources.

\renewcommand*{\arraystretch}{1.5}
\begin{table}[]
    \centering
    \begin{tabular}{|c|c|c|c|}
        \hline
        $a$ & $v_A(a)$ & $v_B(a)$ & $v_C(a)$  \\
        \hline\hline
        1 & 1 & $1+2\epsilon$ & 0 \\
        \hline
        2 & $2-\epsilon$ & $1+\epsilon$ & $2-\epsilon$ \\
        \hline
        3 & 0 & 1 & 1\\
        \hline
    \end{tabular}
    \caption{Valuations of the three agents $A$, $B$, and $C$ with three types of items.
    Here, $\epsilon = 0.07$.}
    \label{tbl:noreachablemea}
\end{table}

\begin{figure}
    \centering
    \includegraphics[scale=0.25]{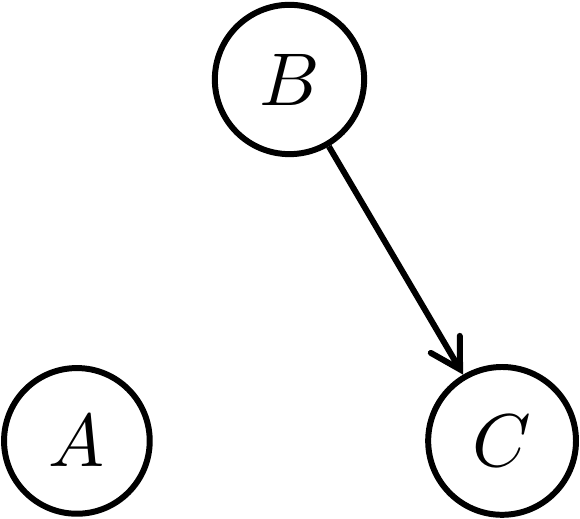}
    \caption{Envy graph of the example discussed in \cref{sec:counterex}.
    Here, $A$ and $B$ are the sources.}
    \label{fig:envygraph}
\end{figure}

We are then left with only one unallocated item $g$ which is of type $2$.
Giving $g$ to any agent does not give an EFX allocation because
\begin{itemize}
    \item if $g$ is given to $A$, then $C$ envies the strict subset of $X_A \uplus \{g\}$ that corresponds to $\b x_A + \b e_2 - \b e_1 = (2,6,0)$,
    
    \item if $g$ is given to $B$, then $A$ envies the strict subset of $X_B \uplus \{g\}$ that corresponds to $\b x_B + \b e_2 - \b e_1 = (5,4,0)$, and
    
    \item if $g$ is given to $C$, then $B$ envies the strict subset of $X_C \uplus \{g\}$ that corresponds to $\b x_C = (0,0,11)$.
\end{itemize}
What about a reachable MEA for a source $s$ in $G_{\mathcal X}$ with respect to $g$? Since there are two sources, there are two possibilities:
\begin{enumerate}
    \item \underline{$s=A$}: $C$ is a MEA of $X_A \uplus \{g\}$ with $\b z_A = \b x_A + \b e_2 - 3\b e_1 = (0,6,0)$.
    However, $B$ is not a MEA of $X_A \uplus \{g\}$ because $B$ does not envy $X_A \uplus \{g\}$.
    And $A$ is not a MEA of $X_A \uplus \{g\}$ because the only $Z_A' \subsetneq X_A \uplus \{g\}$ that $A$ envies corresponds to $\b z_A'=\b x_A + \b e_2 - \b e_1 = (2,6,0)$, but $C$ envies a strict subset of $Z_A'$.
    
    \item \underline{$s=B$}: $A$ is a MEA of $X_B \uplus \{g\}$ with $\b z_B =\b x_B + \b e_2 - \b e_1= (5,4,0)$.
    However, $C$ is not a MEA of $X_B \uplus \{g\}$ because $C$ does not envy $X_B \uplus \{g\}$.
    And $B$ is not a MEA of $X_B \uplus \{g\}$ because $B$ does not envy any strict subset of $X_B \uplus \{g\}$.
\end{enumerate}
Since $A \rightarrow C$ and $B\rightarrow A$ are not edges in $G_{\mathcal X}$, there is no reachable MEA for any source.

Actually, we can say more.
Although $C$ is not a source, we can define a reachable MEA from $C$ with respect to $g$.
If there is one, we can get a strictly Pareto dominant EFX allocation.
In the above example, $B$ is a MEA of $X_C \uplus \{g\}$ with $\b z_C = \b x_C + \b e_2 - 2\b e_3 = (0,1,9)$.
However, $A$ is not a MEA of $X_C \uplus \{g\}$ because $A$ does not envy $X_C \uplus \{g\}$.
And $C$ is not a MEA of $X_C \uplus \{g\}$ because the only $Z_C' \subsetneq X_C \uplus \{g\}$ that $C$ envies corresponds to $\b z_C' = \b x_C + \b e_2 - \b e_1 = (0,1,10)$, but $B$ envies a strict subset of $Z_C'$.
Therefore, there is no reachable MEA even for $C$.

Note that if $g$ had been an item of type $1$ or $3$, then the above counterexample would not have worked.
That is, either adding $g$ to one of the sets would have given an EFX allocation, or there would have been a reachable MEA from one of the sources with respect to $g$.
In either case, we could have obtained a strictly Pareto dominant EFX allocation.

Moreover, the above counterexample does not mean that there is no complete EFX allocation.
Indeed, the complete allocation $\mathcal Y = \langle \b y_A,\b y_B,\b y_C\rangle$, where $\b y_A = (4,3,2)$, $\b y_B = (3,3,4)$, and $\b y_C = (2,3,5)$, is EF, and hence also EFX.

\section*{Acknowledgements}
KM thanks Konstantinos Ameranis and Casey Duckering for informal discussions.
SV acknowledges that the setup of fair allocation of a multiset of items when agents have additive valuations was originally formulated with Madeleine Yang in their final project for the CS238 Optimized Democracy course taught by Ariel Procaccia at Harvard University in Spring 2021.
SV also thanks Ariel Procaccia for helpful discussions.

%\clearpage
%\newpage
\printbibliography

\appendix
\section{Proof of \cref{lem:gcd_result}}\label{app:gcd_result}

In this appendix, we prove \cref{lem:gcd_result}.
\begin{proof}[Proof of \cref{lem:gcd_result}]
Without loss of generality, we can assume that $r=1$; otherwise, we can divide $n_\sfi$'s and $n$ by $r$ so that $\gcd(\frac{n_1}{r},\ldots,\frac{n_d}{r})=1$ and $\sum_{\sfi =1}^d \frac{n_\sfi}{r} = \frac{n}{r}$. Since we assumed that $\gcd(n_1,\ldots,n_d)=1$, by Bézout's lemma, there are integers $Y_1,\ldots,Y_d$ such that
\begin{equation}
    \sum_{\sfi\in[d]} n_\sfi Y_\sfi = q~.\label{eq:bezout}
\end{equation}
Without loss of generality, we can assume that $1\le n_1\le n_2\le \cdots \le n_d$ and $-n < q \le 0$. For any $2\le \sfi \le d$, the transformation $(Y_1,Y_\sfi) \rightarrow (Y_1 \pm n_\sfi,Y_\sfi \mp n_1)$ gives another solution to \eqref{eq:bezout}. Using such transformations, we can ensure that $0\le Y_\sfi < n_1 \le \frac{n}{d}$ for $2\le \sfi\le d$. Then,
\begin{equation*}
    Y_1 = \frac{1}{n_1}\left(q - \sum_{\sfi=2}^d n_\sfi Y_\sfi\right) \in ( -2n,0]~.
\end{equation*}
Our goal now is to increase $Y_1$ so that $|Y_1| \le \frac{7 n}{d}$, while ensuring that the remaining $Y_\sfi$ still satisfy $|Y_\sfi|\le\frac{7 n}{d}$.

Note that $n_\sfi \le \frac{7 n}{d}$ for all $\sfi \le \sfi_0:= \lceil \frac{6d}7 \rceil$. Let us perform the transformations $(Y_1,Y_\sfi) \rightarrow (Y_1 + n_\sfi,Y_\sfi - n_1)$ for any $2\le \sfi \le \sfi_0$ repeatedly while maintaining $Y_\sfi > -\frac{7 n}{d}$, and stop the process if $Y_1$ becomes nonnegative at any step. If the process ends with $Y_1\ge0$, then $Y_1 \le n_\sfi$ for some $2\le\sfi\le\sfi_0$, so $Y_1\le \frac{7 n}{d}$, and we are done. We claim that this is the only way the process can end, i.e., $Y_1$ cannot remain negative throughout the process. Indeed, if $Y_1$ remains negative, then for each $2\le\sfi\le\sfi_0$, we must have performed at least $\frac{6n}{d n_1}$ operations on $Y_\sfi$. This means, $Y_1$ must have increased by at least 
\begin{equation*}
    \sum_{\sfi=2}^{\sfi_0} \left(\frac{6n}{d n_1}\right) n_\sfi \ge (\sfi_0-1)\frac{6n}{d} > 2n~.
\end{equation*}
which is a contradiction because $Y_1>-2n$ at the beginning of the process. Therefore, in the end, $|Y_\sfi| \le \frac{7 n}{d}$ for all $\sfi\in[d]$.
\end{proof}
\end{document}